\documentclass[dvips,generic,preprint,11pt]{imsart}
\setattribute{journal}{name}{}

\renewcommand\baselinestretch{1.5}

\RequirePackage[OT1]{fontenc}
\RequirePackage{amsthm,amsmath}
\RequirePackage[default]{jasa_harvard}
\RequirePackage[colorlinks,citecolor=blue,urlcolor=blue]{hyperref}
\RequirePackage{hypernat}
\RequirePackage{comment}
\RequirePackage{graphicx,latexsym,amssymb}
\RequirePackage{float,epsfig,multirow,rotating,times}
\RequirePackage{upgreek,wrapfig,chngcntr}

\arxiv{math.PR/0000000}
\startlocaldefs
\newcommand {\ctn}{\citeasnoun} 
        
\numberwithin{figure}{section}
\numberwithin{equation}{section}
\theoremstyle{plain}
\newtheorem{theorem}{Theorem}[section]

\newtheorem{remark}{Remark}[section]

\newcommand{\bD}{\boldsymbol{D}}

\newcommand{\bI}{\boldsymbol{I}}

\newcommand{\bS}{\boldsymbol{S}}
\newcommand{\bt}{\boldsymbol{t}}
\newcommand{\bT}{\boldsymbol{T}}
\newcommand{\bE}{\boldsymbol{E}}

\newcommand{\bR}{\boldsymbol{R}}
\newcommand{\bW}{\boldsymbol{W}}

\newcommand{\bx}{\boldsymbol{x}}
\newcommand{\bX}{\boldsymbol{X}}

\newcommand{\bY}{\boldsymbol{Y}}

\newcommand{\bee}{\boldsymbol{e}}

\endlocaldefs

\begin{document}
\renewcommand\baselinestretch{1.5}



\begin{frontmatter}
\title{Uncertainty in Test Score Data and Classically Defined Reliability of Tests $\&$ Test Batteries, using a New Method for Test Dichotomisation}

\runtitle{Reliability using Classical Definition}

\begin{aug}
\author{
{\fnms{Satyendra Nath} \snm{Chakrabartty}\thanksref{t3,m3}\ead[label=e3]{snc12@rediffmail.com}},
{\fnms{Kangrui} \snm{Wang}\thanksref{t4,m4}\ead[label=e4]{kw202@le.ac.uk}},
{\fnms{Dalia} \snm{Chakrabarty}\thanksref{t1,m1,m2}\ead[label=e1]{d.chakrabarty@warwick.ac.uk},\ead[label=e2]{dc252@le.ac.uk}}
},
\thankstext{t3}{Former Director of Indian Maritime University, Kolkata Campus, India} 
\thankstext{t4}{Ph.D student, Department of Mathematics, University of Leicester, U.K.} 
\thankstext{t1}{Lecturer of Statistics at Department of Mathematics,
  University of Leicester and Associate Research fellow at Department of Statistics,
  University of Warwick}

\runauthor{Chakrabartty, Wang $\&$ Chakrabarty}

\affiliation{}

\address{\thanksmark{m3}
Vivek Vihar\\
Noida, India\\
\printead*{e3}
}
\address{\thanksmark{m1} Department of Statistics\\
University of Warwick\\
Coventry CV4 7AL,
U.K.\\
\printead*{e1}\\
\and\\
\thanksmark{m2},\thanksmark{m4}
Department of Mathematics\\
University of Leicester \\
Leicester LE1 7RH,
U.K.\\
\printead*{e2}\\
\printead*{e4}
}

\end{aug}

\begin{abstract} {As with all measurements, the measurement of examinee 
ability, in terms of scores that the examinee obtains in a test, is
also error-ridden. The quantification of such error or uncertainty in
the test score data--or rather the complementary test reliability--is
pursued within the paradigm of Classical Test Theory in a variety of
ways, with no existing method of finding reliability, isomorphic to
the theoretical definition that parametrises reliability as the ratio
of the true score variance and observed (i.e. error-ridden) score
variance. Thus, multiple reliability coefficients for the same test
have been advanced. This paper describes a much needed method of
obtaining reliability of a test as per its theoretical definition, via
a single administration of the test, by using a new fast method of
splitting of a given test into parallel halves, achieving
near-coincident empirical distributions of the two halves. The method
has the desirable property of achieving splitting on the basis of
difficulty of the questions (or items) that constitute the test, thus
allowing for fast computation of reliability even for very large test
data sets, i.e. test data obtained by a very large examinee sample. An
interval estimate for the true score is offered, given an examinee
score, subsequent to the determination of the test reliability. This
method of finding test reliability as per the classical definition can
be extended to find reliability of a set or battery of tests; a method
for determination of the weights implemented in the computation of the
weighted battery score is discussed. We perform empirical illustration
of our method on real and simulated tests, and on a real test battery
comprising two constituent tests. }
\end{abstract}


\end{frontmatter}

\section{Introduction}
Examinee ability is measured by the scores obtained by an examinee in
a test designed to assess such ability. As with all other
measurements, this ability measurement too is fundamentally
uncertain or inclusive of errors. In Classical Test Theory,
quantification of the complementary certainty, or reliability of a
test, is defined as the ratio of the true score variance and the
observed score variance i.e. reliability is defined as the proportion
of observed test score variance that is attributable to true
score. Here, the observed score is treated as inclusive of the
measurement error. This theoretical definition notwithstanding, there
are different methods of obtaining reliability in practice, and
problems arise from the implementation of these different techniques
even for the same test. Importantly, it is to be noted that the
different methods of estimating reliability coefficients differ
differently from the aforementioned classical or theoretical
definition of reliability. This can potentially result in different
estimates of reliability of a particular test even for the same
examinee sample. \ctn{berkowitz_2000} defined reliability as the
degree to which test scores for a group of test takers are consistent
over repeated applications of the test, and are
therefore inferred to be dependable and repeatable for an individual
test taker. In this framework, high reliability means that the test is
trustworthy. \ctn{jacobs_1991} and \ctn{satterly_1994} have opined
that another way to express reliability is in terms of the standard
error of measurement. \ctn{rudner_2002} mentioned that it is
impossible to calculate a reliability coefficient that conforms to the
theoretical definition of the ratio of true score variance and
observed score variance. \ctn{webb_2006} have suggested that the
theoretical reliability coefficient is not practical since true scores
of individuals taking the test are not known. In fact, none of the
existing methods of finding reliability is found to be isomorphic to
the classical definition.

In this paper we present a new methodology to compute the classically
defined reliability coefficient of a multiple-choice binary test,
i.e. a test, responses to the questions of which is either correct or
incorrect, fetching a score of 1 or 0 respectively. Thus, our method
gives the reliability as per the classical definition--thereby
avoiding confusion about different reliability values of a given
test. The method that we advance does not resort to multiple testing,
i.e. administering the same test multiple times to a given examinee
cohort. While avoiding multiple testing, this method has the important
additional benefit of identifying the way of splitting of the test, in
order to ensure that the split halves are equivalent (or parallel), so
that they approach maximal correlation and therefore maximal
split-half reliability. We offer an interval estimate of the true
score of any individual examinee whose obtained score is known.

This paper is organised as follows. In the following section, we
discuss the methods of finding reliability as found in the current
literature. In Section~\ref{sec:batt_rel}, we present a similar
discussion, but this time, in regard to a set or battery of tests. Our method
of finding the reliability as per the classical definition, is
expounded upon in Section~\ref{sec:method}; this is based upon a novel
method of splitting a test into two parallel halves using a new
iterative algorithm that is discussed in Section~\ref{sec:dichot}
along with the benefits of this iterative scheme. Subsequently, we
proceed to discuss the computation of the reliability of a battery of
tests using summative scores and weighted scores
(Section~\ref{sec:weights}). Empirical verification of the presented
methodology is undertaken using four sets of simulated data (in
Section~\ref{sec:simulated}). Implementation of the method to real data
is discussed in Section~\ref{sec:real}. We round up the paper in
Section~\ref{sec:conclu}, by recalling the salient outcomes of the
work.

\section{Currently available methods}
\label{sec:lit}
In the multiple testing (or the test-retest) approach, the main
concern is that the correlation between test scores coming from
distinct administrations, depends on the time gap between the
administrations, since the sample may learn or forget something in
between and also get acquainted to the test. Thus, different values of
reliability can be obtained depending on the time gap
\cite{meadows_2005}. Also, test-retest reliability primarily reflects
stability of scores and may vary depending on the homogeneity of
groups of examinees \cite{gualtieri_2006}.

In the split half method,
the test is divided into two sub-tests, i.e. the whole set of questions is divided into two sets of questions and the entire instrument is administered to one
examinee sample. The split-half reliability estimate is the
correlation between scores of these two parallel halves, though
researchers like \ctn{kaplan_2001} recommend finding
reliability of the entire test using the Spearman-Brown formula. Value
of the split-half reliability depends on how the test is
dichotomised. We acknowledge that a test consisting of 2$\eta$-number
of items can be split in half in $\displaystyle{2\eta\choose\eta}$
number of ways and that each method of splitting will
imply a distinct value of reliability in general. In this context, it
merits mention that it is of crucial importance to identify that way
of splitting that ensures that the split halves are parallel and
the split-half reliability is maximum. 

In attempts to compute reliability based on ``internal consistency'',
the aim is to quantify how well the questions in the test, or test
``items'', reflect the same construct. These include the usage of
Cronbach's Alpha ($\alpha$): the average of all possible split-half
estimates of reliability, computed as a function of the ratio of the
total of the variance of the items and variance of the examinee
scores. Here, the ``item score'' is the total of the scores obtained
in that item by all the examinees who are taking the test, while an
examinee score is the total score of such an examinee, across all the
items in the test.
$\alpha$ can be shown to
provide a lower bound for the reliability under certain
assumptions. Limitations of this method have been reported by
\ctn{eisinga_2012} and \ctn{ritter_2010}. \ctn{panayides_2013}
observed that 
high values of $\alpha$ do not necessarily mean higher reliability
and better quality of scales or tests. In fact very high values of
$\alpha$ could be indicating lengthy scales, a narrow coverage of the
construct under consideration and/or parallel items
\cite{boyle_1991,kline_1979} or redundancy of items 
\cite{streiner_2003}.

In another approach, referred to as the ``parallel forms'' approach,
the test is divided into two forms or sub-tests by randomly selecting
the test items to comprise either sub-test, under the assumption that
the randomly divided halves are parallel or equivalent. Such a demand
of parallelity requires generating lots of items that reflect the same
construct, which is difficult to achieve in practice. This approach is
very similar to the split-half reliability described above, except
that the parallel forms constructed herein, can be used independently
of each other and are considered equivalent. In practice, this
assumption of strict parallelity is too restrictive.

Each of the above method of estimation of reliability has certain
advantages and disadvantages. However, the estimation of the
reliability coefficient under each such method deviates differently
from the theoretical definition and consequently, gives different
values for reliability of a single test. In general, test-retest
reliability with significant time gap is lower in value than the
parallel forms reliability and reliability computed under the model of
internal consistency amongst the test items. 
To summarise, there is a potency for confusion over the
trustworthiness of a test, emanating out of the inconsistencies
amongst the different available methods that are implemented to
compute reliability coefficients. This state of affairs motivates a
need to estimate reliability in an easily reproducible way, from a
single administration of the test, using the theoretical definition of
reliability i.e. as a ratio of true score variance and observed score
variance.

\subsection{Battery reliability}
\label{sec:batt_rel}
\noindent
\ctn{bock_1966} derived reliability of a battery using Multivariate
Analysis of Variance (MANOVA), where $r$ number of parallel forms of a
test battery were given to each individual in the examinee sample. It
can be proved that average canonical reliabilities from MANOVA
coincides with the canonical reliability for the Mahalanobis distant
function.
\ctn{wood_1987} compared three multivariate models
(canonical reliability model, maximum generalisability model,
canonical correlation model) for estimating reliability of test
battery and observed that the maximum generalisability model showed
the least degree of bias, smallest errors in estimation, and the
greatest relative efficiency across all experimental
conditions. \ctn{ogasawara_2009} considered four estimators of the
reliability of a composite score based on a factor analysis approach,
and five estimators of the maximal reliability for the composite
scores and found that the same estimators are obtained when Wishart
Maximum Likelihood estimators are used. \ctn{conger_1973}
considered canonical reliability as the ratio of the average squared
distance among the true scores to the average squared distance among
the observed scores. They thereby showed that the canonical
reliability is consistent with multivariate analogues of parallel
forms of correlation, squared correlation between true and observed
scores and an analysis of variance formulation. They opined that the
factors corresponding to the small eigenvalues might show significant
differences and should not be discarded. Equivalences amongst
canonical factor analysis, canonical reliability and principal
component analysis were studied by \ctn{conger_1976}; they
found that if variables are scaled so that they have equal measurement
errors, then canonical factor analysis on all non-error variance,
principal component analysis and canonical reliability analysis, give
rise to equivalent results.

In this paper, the objective is to present a methodology for obtaining
reliability as per its theoretical definition from a single
administration of a test, and to extend this methodology to find ways
of obtaining reliability of a battery of tests.

\section{Methodology}
\label{sec:method}
\noindent
Consider a test consisting of $n$ items, administered to $N$
persons. The test score of the $i$-th examinee is $X_i$,
$i=1,\ldots,N$ and $\bX =(X_1, X_2, \ldots, X_N)^T$ is referred to as
the test score vector. The vector $\bI = (I_1, I_2, \ldots,I_N)^T$
depicts the maximum possible scores for the test where $I_i = n$
$\forall i =1,2,\ldots,N$. Here $\bX,\bI\in{\cal X}\subset {\mathbb
  R}^N$ where we refer to ${\cal X}$ as the $N$-dimensional person
space. Sorting the components of $\bX$ in increasing order will
produce a merit list for thesample of examinees, in terms of the
observed scores and thereby help infer the ability distribution of the
sample of examinees. However, parameters of the test can be obtained
primarily in terms of the 2-norm of the vectors $\bX$ and $bI$,
i.e. $\parallel\bX\parallel$, $\parallel\bI\parallel$, and the angle
$\theta_X$ between these two vectors.

Mean of the test scores ${\bar{X}}$ is given as follows. Since
$\cos\theta_X = \displaystyle{\frac{\sum_i^N X_i I_i}{\parallel\bX\parallel \parallel\bI\parallel}}\Longrightarrow \displaystyle{ {\sum_i^N X_i}}/N = 
\displaystyle{\frac{\parallel\bX\parallel \cos\theta_X}{\sqrt{N}}}={\bar{X}}$.
Similarly, mean of the true score vector $\bT$ is
${\bar{T}} = \displaystyle{\frac{\parallel\bT\parallel \cos\theta_T}{\sqrt{N}}}$ where $\theta_T$ is the angle
between the true score vector and the ideal score vector $\bI$.

Test variance $S_X^2$ can also be
obtained in this framework.  For $N$ persons who take the test,
$S_X^2 = \displaystyle{\frac{1}{N}\parallel\bx\parallel^2}
\Longrightarrow S_X = \displaystyle{\frac{1}{\sqrt{N}}\parallel\bx\parallel}$ 
where $x_i=\displaystyle{X_i-{\bar{X}}}$ is the $i$-th deviation
score, $i=1,2,\ldots,N$. The
relationship between the norm of deviation score and norm of test score can be easily derived as
$\parallel\bx\parallel^2 = \parallel\bX\parallel^2 \sin^2\theta_X$
Similarly, 
$\parallel\bt\parallel^2 = \parallel\bT\parallel^2 \sin^2\theta_T$
where $t_i=\displaystyle{T_i-{\bar{T}}}$.

Using this relationship between deviation and test scores, test variance can be
rephrased as
\begin{equation}
S_X^2 = \displaystyle{\frac{\parallel \bX\parallel^2\sin^2\theta_X}{N}}.\quad{\textrm{Also}}\quad
S_T^2 = \displaystyle{\frac{\parallel \bT\parallel^2\sin^2\theta_T}{N}}.
\label{eqn:variance}
\end{equation}

The aforementioned relationship between test and deviation scores gives
$\sin^2\theta_X = \displaystyle{\frac{\parallel\bx\parallel^2}{\parallel\bX\parallel^2}}$ and
$\sin^2\theta_T = \displaystyle{\frac{\parallel\bt\parallel^2}{\parallel\bT\parallel^2}}$.  		                        

At the same time, in Classical Test Theory, the means of the observed
and true scores coincide, i.e. ${\bar{X}}={\bar{T}}$, as in this
framework, the observed score is represented as the sum of true score
(obtained under the paradigm of no errors) and the error, where the
error is assumed to be distributed as a normal with zero mean and
variance referred to as the error variance. Then recalling the
definition of these sample means, we get
\begin{equation}
\cos\theta_T = \displaystyle{\frac{\parallel
    X\parallel}{\parallel T\parallel} \cos\theta_X}.
\label{eqn:theta_T_X}
\end{equation}

This gives the relationship amongst $\parallel\bX\parallel$,
$\parallel\bT\parallel$, $\theta_T$, $\theta_X$.  The stage is
now set for us to develop the methodology for computing reliability
using the classical definition.

\section{Our method}

\subsection{Background}
\label{sec:reliability}
\noindent
Reliability, $r_{tt}$, of a test is defined clasically as $r_{tt}=\displaystyle{\frac{S_T^2}{S_X^2}}$.  
Thus, to know $r_{tt}$  one needs to know the value of the true score variance or the error variance.
For this, let us concentrate on parallel tests. As per the classical definition, two tests ``$g$'' and ``$h$'' are parallel if 	
\begin{equation}
T_{i}^{(g)}  =  T_{i}^{(h)}\quad{\textrm{and}}\quad   S_{e}^{(g)} = S_{e}^{(h)}\quad i=1,2,\ldots,N  
\label{eqn:ssaresame}       
\end{equation}  		
where the superscript ``$g$'' refers to sub-test $g$ and superscript ``$h$'' to sub-test $h$, and $S_{e}^{(p)}$ is the standard deviation of the
error scores in the $p$-th sub-test, $p=g,h$. Thus, $T_i^{(g)}$ is the true score of the $i$-th examinee in the $g$-th test and $T_i^{(h)}$ that in the $h$-th test. Similarly, the observed scores of the $i$-th examinee, in the two tests, are $X_i^{(g)}$ and $X_i^{(g)}$.

From here on, we interpret ``g'' and ``h'' as the two parallel
sub-tests that a given test is dichotomised into. This implies that
the observed score vectors of these two parallel tests can be
represented by two points $\bX_g$ and $\bX_h$, in ${\cal X}$ such that
$\bX_g = \bT_g + \bE_g$ and $\bX_h = \bT_h + \bE_h$ in the paradigm of
Classical Test Theory, where $\bE_p$ is the error score in the $p$-th
test, $p=g, h$. (Here $\bX_p=(X_1^{(p)},\ldots,X_N^{(p)})$,
$p=g,h$). Now, recalling that for parallel tests $g$ and $h$,
$T_i^{(g)}=T_i^{(h)}\forall i=1,\ldots,N\Longrightarrow
\bT_g=\bT_h\Longrightarrow \bX_g-\bX_h = \bE_g-\bE_h$, so that
\begin{equation}
\parallel \bX_g\parallel^2 + \parallel \bX_h\parallel^2 - 2\parallel \bX_g\parallel \parallel \bX_h\parallel\cos\theta_{gh} = \parallel \bE_g\parallel^2 + \parallel \bE_h\parallel^2 - 2\parallel \bE_g\parallel \parallel \bE_h\parallel\cos\theta^{(E)}_{gh}
\label{eqn:1sstep}
\end{equation}     
where $\theta_{gh}$ is the angle between $\bX_g$ and $\bX_h$ while
$\theta^{(E)}_{gh}$ is the angle between $\bE_g$ and $\bE_h$. But, the correlation of the errors in two parallel tests is zero (follows from equality of standard deviation of error sores in the parallel sub-tests and equality of means of the error scores--see Equation~\ref{eqn:theta_T_X}). The geometrical interpretation of this is that the error vectors of the two parallel tests are orthogonal, i.e. $\cos\theta^{(E)}_{gh}=0$. Then Equation~\ref{eqn:1sstep} can be written as\\
\vspace{-.5cm}
\begin{center}
$\parallel \bX_g\parallel^2 + \parallel \bX_h\parallel^2 - 2\parallel \bX_g\parallel \parallel \bX_h\parallel\cos\theta_{gh}$
\end{center}
\begin{equation}
=\parallel \bE_g\parallel^2 + \parallel \bE_h\parallel^2 = N(S_e^{(g)})^2 + N(S_e^{(h)})^2 = 2N(S_e^{(g)})^2 
\label{eqn:2ndstep}
\end{equation}     
where we have used equality of error variances of parallel tests in
the last step. Equations~\ref{eqn:2ndstep} can be employed to find the
value of $S_e^{(g)}$ from the data (on observed scores). In other words, we can use the
available test score data in this equation to achieve the error
variance of either parallel sub-test that a test can be dichotomised
into. Alternatively, if two parallel tests exist, then
Equations~\ref{eqn:2ndstep} can be invoked to compute the error
variance in either of the two parallel tests. 


Now, Equations~\ref{eqn:2ndstep} suggest that the error variance $(S_e^{(test)})^2$ of the entire test is 
\begin{equation}
(S_e^{(test)})^2 = 2(S_e^{(g)})^2 = \displaystyle{\frac{\parallel \bX_g\parallel^2 + \parallel \bX_h\parallel^2 - 2\parallel \bX_g\parallel \parallel \bX_h\parallel\cos\theta_{gh}}{N}}
\label{eqn:finally}
\end{equation}
Then recalling that in the paradigm of Classical Test Theory, the true score is by definition independent of the error, it follows that the observed score variance $S_X^2$ is sum of true score variance $S_T^2$ and error variance $(S_e^{(test)})^2$ the classical definition of reliability gives
\vspace{-.2cm}
\begin{center}
$r_{tt} = \displaystyle{\frac{S_T^2}{S_X^2}} = \displaystyle{1 - \frac{(S_e^{(test)})^2}{S_X^2}}= $
\end{center}
\begin{equation}
1 - \displaystyle{\frac{\parallel \bX_g\parallel^2 + \parallel \bX_h\parallel^2 - 2\parallel \bX_g\parallel \parallel \bX_h\parallel\cos\theta_{gh}}{N S_X^2}} 
= 1 - \displaystyle{\frac{\parallel \bX_g\parallel^2 + \parallel \bX_h\parallel^2 - 2\sum_{i=1}^N X_i^{(g)}X_i^{(h)}}{N S_X^2}} 
\label{eqn:3rdstep}
\end{equation}
As $\parallel\bX_g\parallel \parallel\bX_h\parallel\cos\theta_{gh} = \displaystyle{\sum_{i=1}^N X_i^{(g)}X_i^{(h)}}$, Equation~\ref{eqn:3rdstep} can be simplified to give
\begin{equation}
r_{tt} = 1 - \displaystyle{\frac{2\parallel X_g\parallel^2 - 2\sum_{i=1}^N X_i^{(g)}X_i^{(h)}}{N S_X^2} }
\label{eqn:4thstep}
\end{equation}
since $\parallel\bX_g\parallel= \parallel\bX_h\parallel$ for parallel
tests, given that ${\bar{X}}_g={\bar{X}}_h$ and $S_e^{(g)}=
S_e^{(h)}$.  Equation~\ref{eqn:3rdstep} and Equation~\ref{eqn:4thstep}
give a unique way of finding reliability of a test from a single
administration, using the classical definition of reliability as long
as dichotomisation of the test is performed into two parallel
sub-tests.  Importantly, in this framework, the hypothesis of equality
of error variances of the two halves of a given test can be tested
using an F-test. 
In addition, the method also provides a way to
estimate true scores from the data. We discuss this
later in this section.

\subsection{Proposed method of split-half}
\label{sec:dichot}
\noindent
\ctn{chakrabartty_2011} gave a method for splitting a test into 2
parallel halves. Here we give a novel method of splitting a test into
2 parallel halves--$g$ and $h$--that have nearly equal means and
variances of the observed scores. The splitting is initiated by the
determination of the item-wise total score for each item. So let the
$j$-th item in the test have the item-wise score
$\tau_j:=\displaystyle{\sum\limits_{i=1}^N X_i^{(j)}}$, where
$X_i^{(j)}$ is the $i$-th examinee's score in the $j$-th item,
$j=1,\ldots,n$. Our method of splitting is as follows.

\begin{enumerate}
\item[Step-I] The item-wise scores are sorted in an ascending order
  resulting in the ordered sequence
  $\tau_1,\tau_2,\ldots,\tau_n$. Following this, the item with the
  highest total score is identified and allocated to the $g$-th
  sub-test. The item with second highest total score is then allocated
  to the $h$-th test, while the item with the third highest score is
  assigned to $h$-th test and the fourth highest to the $g$-th test,
  and so on. In other words, allocation of items is performed to
  ensure realisation of the following structure.
\begin{center}{
\begin{tabular}{l l l }
{\underline{sub-test $g$}} &
 {\underline{sub-test $h$}} & {\underline{difference in item-wise scores of 2 sub-test}} \\
$\tau_1$ & $\tau_2$ & $\tau_1 - \tau_2 \geq 0$  \\
$\tau_4$ & $\tau_3$ & $\tau_4 - \tau_3 \leq 0$ \\
$\vdots$ & $\vdots$ & $\vdots$ \\
\end{tabular}}
\end{center}
where we assume $n$ to be even; for tests with an odd number of items,
we ignore the last item for the purposes of dichotomisation. The
sub-tests obtained after the very first slotting of the sequence
$\{\tau_k\}_{k=1}^n$ into the sub-tests, following this suggested
method of distribution, is referred to as the ``seed sub-tests''.
\item[Step-II] Next, the difference of item-wise scores in every item
  of the $g$-th and $h$-th sub-tests is recorded and the sum ${\cal
    S}$ of these differences is computed (total of column 3
  in the above table). If the value of ${\cal S}$ is zero,
  we terminate the process of distribution of items across the 2
  sub-tests, otherwise we proceed to the next step.
\item[Step-III] We identify rows in the above table, the swapping of
  the entries of columns 1 and 2 of which, results in the reduction of
  $\vert{\cal S}\vert$, where $\vert\cdot\vert$ denotes absolute
  value. Let the row numbers of such rows be $\rho^{(\star_\ell)}$,
  $\ell=1,2,\ldots,n^{(\star)}$ where $n^{(\star)}\leq n/2$. We swap the
  $\rho^{(\star_\ell)}$-th item of the $g$-th sub-test with the
  $\rho^{(\star_\ell)}$-th item of the $h$-th sub-test and
  re-calculate sum of the entries of the revised $g$-th sub-test and
  $h$-th sub-test. If the revised value of $\vert{\cal S}\vert$ is
  zero or a number close to zero that does not reduce upon further
  iterations, we stop the iteration; otherwise we return to the
  identification of the row numbers $\rho^{(\star_\ell)}$ and proceed
  therefrom again.
\end{enumerate}

We considered other methods of splitting of the test as well,
including methods in which swapping of items of the two sub-tests is
allowed not just for a given row, but also across rows, and
minimisation of the sum of the differences between the items in the
$g$-th and $h$-th sub-tests is not the only criterion. To be precise,
we conduct methods of dichotomisation in which in Step-I we compute
the sum ${\cal S}$ of the differences between the item scores in the 2
sub-tests, as well as consider the sum ${\cal S}_{sq}$ of differences
of squares of the item scores in the $g$-th and $h$-th sub-tests. We
then identify row numbers $\rho^{(\star_1)}$ and $\rho^{(\star_2)}$ in
the above table, such that swapping the item in row $\rho^{(\star_1)}$
of the $g$-th sub-test with the entry in row $\rho^{(\star_2)}$ of the
$h$-th sub-test results in a reduction of $\vert {\cal S}\vert\times
\vert{\cal S}_{sq}\vert$; $\rho^{(\star_1)}, \rho^{(\star_1)}
\leq n/2$. When this product can no longer be reduced over a chosen
number of iterations, the scheme is stopped, otherwise the search for
the parallel halves that result in the minimum value of this product,
is maintained. However, parallelisation obtained with this method of
splitting that seeks to minimise $\vert {\cal S}\vert\times\vert{\cal
  S}_{sq}\vert$ was empirically found to yield similar results as with
the method that seeks to minimise $\vert {\cal S}\vert$. Here by
``similar results'' is implied dichotomisation of the given test into
two sub-tests, the means and variances of which are equally close to
each other in both methods, i.e. sub-tests are approximately equally
parallel.  The reason for such empirically noted similarity is
discussed in the following section. Given this, we advance the method
enumerated above in Steps-I to III, as our chosen method of
dichotomisation.

By the mean (or variance) of a sub-test, is implied mean (or variance)
of the vector of examinee scores in the $n/2$ items that comprise that
sub-test, i.e. the mean (or variance) of the score vector of the
sub-test. Thus, if the item numbers $g_1,\ldots,g_{n/2}$ comprise the
$g$-th sub-test, $g_j\in\{1,2,\ldots,n\},\:j=1,\ldots,n/2$, then the
$N$-dimensional score vector of the $g$-th sub-test is
$\bX^{(g)}:=(X_1^{(g)}, \ldots, X_N^{(g)})^T$, where
$X_i^{(g)}=\displaystyle{\sum\limits_{j=1}^{{n/2}} X_i^{(g_j)}}$, i.e. the score acieved by the $i$-th examinee across all the items that are included in the $g$-th sub-test; this is to be distinguished from $\tau_j$--the score in the $j$-th item, summed over all examinees. Given that $X_i^{(j)}$ is either 0 or 1, we get that $X_i^{(g)}\leq {n/2}$ and $\tau_j\leq N,\:\forall\:i, j$. We
similarly define the score-vector $\bX^{(h)}$ of the $h$-th sub-test.

Also, in he following section we identify score in the $j$-th item
that is constituent of the $g$-th sub-test as $\tau_j^{(g)},
j=1\ldots,n/2$, Similarly we define $\tau_j^{(h)}$.

\begin{remark}
\label{remark:remark1}
The order of our algorithm is independent of the examinee number and driven by the number of items in each sub-test that the test of $n$ items is dichotomised into; the order is ${\cal O}((n/2)^2)$. 
\end{remark}

\subsection{Benefits of our method of parallelisation}
\label{sec:benefits_parallel}
\noindent 


\begin{theorem}
\label{theorem:means}
Difference between means of $g$-th and $h$-th sub-tests is minimised.
\end{theorem}
\begin{proof}
Let the score vectors of the $g$-th and $h$-th sub-tests be
$\bX^{(g)}:=(X_1^{(g)}, \ldots, X_N^{(g)})^T$ and 
$\bX^{(h)}:=(X_1^{(h)}, \ldots, X_N^{(h)})^T$, 
where
$X_i^{(\cdot)}=\displaystyle{\sum\limits_{j=1}^{{n/2}} X_i^{(\cdot_j)}}.$
Now, mean of the $g$-th sub-test is 
${\bar{X}}_g=\displaystyle{\frac{\sum\limits_{i=1}^{N} X_i^{(g)}}{N}}$ 
and mean of the $h$-th sub-test is 
${\bar{X}}_h=\displaystyle{\frac{\sum\limits_{i=1}^{N} X_i^{(h)}}{N}}$.

Let item score of the $j$-th item in the $g$-th sub-test be $\tau_j^{(g)}$, $j=1,\ldots,n/2$. Similarly, let score of $j$-th item in the $h$-th sub-test be $\tau_j^{(h)}$.

But, sum of item scores in all $n/2$ items in a sub-test, is equal to sum of examinee scores achieved in these $n/2$ items, i.e.
\begin{eqnarray}
\displaystyle{\sum\limits_{i=1}^{N} X_i^{(g)}} &= & \displaystyle{\sum\limits_{j=1}^{n/2} \tau_j^{(g)}}\nonumber\\
\displaystyle{\sum\limits_{i=1}^{N} X_i^{(h)}} &= & \displaystyle{\sum\limits_{j=1}^{n/2} \tau_j^{(h)}}
\label{eqn:item-ex}
\end{eqnarray}

Now, by definition, 
\begin{eqnarray}
\vert{\cal S}\vert &=& \displaystyle{{\Big\vert} \left[\tau_1^{(g)}-\tau_1^{(h)}\right] + \ldots + \left[\tau_{n/2}^{(g)}-\tau_{n/2}^{(h)}\right]{\Big\vert}} \\ \nonumber 
         &=& \displaystyle{{\Big\vert}\left[\tau_1^{(g)}+\ldots+\tau_{n/2}^{(g)}\right] - 
             \left[\tau_1^{(h)}+\ldots+\tau_{n/2}^{(h)}\right]{\Big\vert}}. \nonumber
\end{eqnarray}  
At the end of the splitting of the test, let $\vert{\cal S}\vert=\epsilon$. Then $\epsilon$ is the minimum value of $\vert{\cal S}\vert$ by our method.
 
Then using 
$$\displaystyle{{\Big\vert}\left[\tau_1^{(g)}+\ldots+\tau_{n/2}^{(g)}\right] - 
             \left[\tau_1^{(h)}+\ldots+\tau_{n/2}^{(h)}\right]{\Big\vert}} =\epsilon$$
in Equations~\ref{eqn:item-ex}, we get
\begin{equation}
\displaystyle{{\Big\vert}\sum\limits_{i=1}^{N} X_i^{(g)} - \sum\limits_{i=1}^{N} X_i^{(h)}{\Big\vert}} 
= \displaystyle{{\Big\vert}\sum\limits_{j=1}^{n/2} \tau_j^{(g)} - \sum\limits_{j=1}^{n/2} \tau_j^{(h)}{\Big\vert}} = \epsilon
\label{eqn:inside}
\end{equation}

Then ${\Big\vert}{\bar{X}}_g-{\bar{X}}_h{\Big\vert} = \displaystyle{{\Bigg\vert}\frac{\sum\limits_{i=1}^{N} X_i^{(g)}}{N} -\frac{\sum\limits_{i=1}^{N} X_i^{(h)}}{N}{\Bigg\vert}}= \epsilon/N$ (using Equation~\ref{eqn:inside}), i.e. difference between means is minimised by our method. 
\end{proof}

\vspace{1cm}
\begin{remark}
In our work, by ``near-equal means'', we imply means of the sub-tests, the difference between which is minimised using our method; typically, this minimum value is close to 0. Thus, the means of the two sub-tests, are near-equal.
\end{remark}

\vspace{1cm}
\begin{theorem}
\label{theorem:vars}
Absolute difference between sums of squares of examinee scores in the $g$-th and $h$-th sub-tests is of the order of $\epsilon^2$, if absolute difference between sums of examinee scores is $\epsilon$.
\end{theorem}
\begin{proof}
Absolute difference between sums of scores $g$-th and $h$-th sub-tests
is minimised in our method (Equation~\ref{eqn:inside}), with
\begin{equation}
\displaystyle{\sum\limits_{i=1}^{N} X_i^{(g)}}= \displaystyle{\sum\limits_{i=1}^{N} X_i^{(h)}}\pm\epsilon. \nonumber
\end{equation} 
For our method, $\epsilon$ is small. Thus we state:
\begin{equation}
\displaystyle{\sum\limits_{i=1}^{N} X_i^{(g)}} \approx \displaystyle{\sum\limits_{i=1}^{N} X_i^{(h)}},
\label{eqn:basic}
\end{equation} 
so that $T:=\displaystyle{\sum\limits_{i=1}^{N} X_i^{(g)}}\Longrightarrow \displaystyle{\sum\limits_{i=1}^{N} X_i^{(h)}}=T \mp \epsilon.$

Here we define
$$X_i^{(g)}=\displaystyle{\sum\limits_{j=1}^{n/2} X_i^{(g_j)}}$$
and similarly, we define $X_i^{(h)}$, $\forall i=1,\ldots,N$.

Now,
\begin{eqnarray}
\displaystyle{\left(\sum\limits_{i=1}^{N} X_i^{(g)}\right)^2 } 
&=& \displaystyle{\sum\limits_{i=1}^{N}\left(X_i^{(g)}\right)^2}+\nonumber \\
&& \displaystyle{X_1^{(g)}X_2^{(g)} + \ldots + X_1^{(g)}X_N^{(g)} +} \nonumber \\
&& \ldots + \nonumber \\
&& \displaystyle{X_N^{(g)}X_1^{(g)} + \ldots + X_N^{(g)}X_{N-1}^{(g)}} \nonumber \\
&=& \displaystyle{\sum\limits_{i=1}^{N}\left(X_i^{(g)}\right)^2+
\sum\limits_{i=1}^{N}\left[\sum\limits_{k=1; k\neq i}^{N} X_i^{(g)}X_k^{(g)}\right]} \nonumber \\
\label{eqn:1}
\end{eqnarray}

Now, for $k\neq i$,
\begin{eqnarray}
X_i^{(g)}X_k^{(g)}&=& \displaystyle{\left(X_i^{(g_1)}+\ldots+X_i^{(g_{n/2})}\right)\left(X_k^{(g_1)}+\ldots+X_k^{(g_{n/2})}\right)}\nonumber \\
&=& \displaystyle{X_i^{(g_1)}X_k^{(g_1)} + X_i^{(g_1)}X_k^{(g_2)} +\ldots + X_i^{(g_1)}X_{k}^{(g_{n/2})}+} \nonumber \\
&&\ldots +\nonumber \\
&& \displaystyle{X_i^{(g_{n/2})}X_k^{(g_1)} + X_i^{(g_{n/2})}X_k^{(g_2)} +\ldots + X_i^{(g_{n/2})}X_{k}^{(g_{n/2})}+} \nonumber \\
&=& \displaystyle{\sum\limits_{j=1}^{n/2}\left[\sum\limits_{j'=1}^{n/2}\left(\mbox{number of times each of $X_i^{(j)}$ and $X_k^{(j')}$ is 1}\right)\right]}\nonumber \\
&=& \displaystyle{\left[\sum\limits_{j=1}^{n/2}\left(\mbox{number of times $X_i^{(g_j)}=1$}\right)\right] \left[\sum\limits_{j=1}^{n/2}\left(\mbox{number of times $X_k^{(g_j)}=1$}\right)\right]}\nonumber \\
&\approx& (n/2)^2 \displaystyle{\left[\sum\limits_{j=1}^{n/2} \Pr(X_i^{(g_j)}=1)\right]\left[\sum\limits_{j=1}^{n/2} \Pr(X_k^{(g_j)}=1)\right]}
\label{eqn:2}
\end{eqnarray}
Here $X_i^{(g_j)}$ is the score of the $i$-th examinee in the $j$-th item of the $g$-th sub-test and can attain values of either 1 or 0, with probability $p_i^{(g_j)}$ or 1-$p_i^{(g_j)}$ respectively. Thus,
$X_i^{(g_j)}\sim{\textrm{Bernoulli}}(p_i^{(g_j)})$, i.e. $\Pr(X_i^{(g_j)}=1)=p_i^{(g_j)}$. The approximation in Equation~\ref{eqn:2} stems from approximating the probability for an event, with its relative frequency.
Then following Equation~\ref{eqn:2}, we get
\begin{eqnarray}
\displaystyle{\sum\limits_{i=1}^{N}\sum\limits_{k=1; k\neq i}^{N} X_i^{(g)}X_k^{(g)}}
&\approx & (n/2)^2 \displaystyle{\sum\limits_{i=1}^{N}\sum\limits_{k=1; k\neq i}^{N}
\left[\sum\limits_{j=1}^{n/2} p_i^{(g_j)}\sum\limits_{j=1}^{n/2} p_k^{(g_j)}\right]}\nonumber \\
&=& (n/2)^2 \displaystyle{\sum\limits_{i=1}^{N}\left[
\left(\sum\limits_{j=1}^{n/2} p_i^{(g_j)}\right)\left(\sum\limits_{k=1; k\neq i}^{N}\sum\limits_{j=1}^{n/2} p_k^{(g_j)}\right)\right]}
\label{eqn:3}
\end{eqnarray}
Now, Equation~\ref{eqn:basic} implies that 
\begin{eqnarray}
\displaystyle{\sum\limits_{i=1}^{N}\left[\sum\limits_{j=1}^{n/2} X_i^{(g_j)}\right]} &\approx& \displaystyle{\sum\limits_{i=1}^{N}\left[\sum\limits_{j=1}^{n/2} X_i^{(h_j)}\right]}\quad{\mbox{i.e.}}\nonumber \\
\displaystyle{\sum\limits_{i=1}^{N}\left[\sum\limits_{j=1}^{n/2} \left(\mbox{number of times $X_i^{(g_j)}=1$}\right)\right]} &\approx& \displaystyle{\sum\limits_{i=1}^{N}\left[\sum\limits_{j=1}^{n/2} \left(\mbox{number of times $X_i^{(h_j)}=1$}\right)\right]}\quad{\mbox{i.e.}}\nonumber \\
\displaystyle{\sum\limits_{i=1}^{N}\left[\sum\limits_{j=1}^{n/2} p_i^{(g_j)}\right]} &\approx& \displaystyle{\sum\limits_{i=1}^{N}\left[\sum\limits_{j=1}^{n/2} p_i^{(h_j)}\right]} 
\label{eqn:4}
\end{eqnarray}
Then if we delete any 1 out of the $N$ examinees, over which the outer summation on the RHS and LHS of the last approximate equality is carried out, it is expected that the approximation expressed in statement~\ref{eqn:4} would still be valid. This is especially the case if $N$ is large. In other words, bigger the $N$, smaller is the distortion affected on the structure of the sub-tests generaetd by splitting the test data obtained after deleting the score of any 1 of the $N$ examinees from the original full test data. Then using statement~\ref{eqn:4} for a large $N$, we can write
\begin{equation}
\displaystyle{\sum\limits_{k=1, k\neq i}^{N}\left[\sum\limits_{j=1}^{n/2} p_k^{(g_j)}\right]} \approx \displaystyle{\sum\limits_{k=1;k\neq i}^{N}\left[\sum\limits_{j=1}^{n/2} p_k^{(h_j)}\right]}
\label{eqn:5}
\end{equation}
where  
$$\displaystyle{\sum\limits_{i=1}^{N}\left[\sum\limits_{j=1}^{n/2} p_i^{(g_j)}\right]}= \displaystyle{\sum\limits_{i=1}^{N}\left[\sum\limits_{j=1}^{n/2} p_i^{(h_j)}\right] \pm \epsilon'},$$
follows from Equation~\ref{eqn:basic} that suggests that $\displaystyle{\sum\limits_{i=1}^{N} X_i^{(g)}} = \displaystyle{\sum\limits_{i=1}^{N} X_i^{(h)}}\pm \epsilon$. Thus $\epsilon\in{\mathbb Z}_{\geq 0}$ and $\epsilon'\in{\mathbb R}_{\geq 0}$ such that $\epsilon \geq \epsilon'$, given that $\epsilon'$ is the absolute difference between sums of probability of correct response in the $g$-th sub-test and that in the $h$-th sub-test while $\epsilon$ is the absolute difference between sums of scores in the two sub-tests.

In other words, if we define the sum of probabilities of correct
response in the $g$-th sub-test, $T'$, as
$$T' := \displaystyle{\sum\limits_{k=1, k\neq i}^{N}\left[\sum\limits_{j=1}^{n/2} p_k^{(g_j)}\right]},$$ then
\begin{equation}
\displaystyle{\sum\limits_{k=1, k\neq i}^{N}\left[\sum\limits_{j=1}^{n/2} p_k^{(h_j)}\right]}\approx T'\mp \epsilon'.
\label{eqn:7}
\end{equation}

Using this, for sub-test $g$, in the last line of Equations~\ref{eqn:3} we get
\begin{equation}
\displaystyle{\sum\limits_{i=1}^{N}\sum\limits_{k=1; k\neq i}^{N} X_i^{(g)}X_k^{(g)}} = (n/2)^2 T'\displaystyle{\sum\limits_{i=1}^{N}\left[
\left(\sum\limits_{j=1}^{n/2} p_i^{(g_j)}\right)\right]}.
\label{eqn:6}
\end{equation}
For sub-test $h$, 
$$
\displaystyle{\sum\limits_{i=1}^{N}\sum\limits_{k=1; k\neq i}^{N} X_i^{(h)}X_k^{(h)}} \approx (n/2)^2 T'\displaystyle{\sum\limits_{i=1}^{N}\left[
\left(\sum\limits_{j=1}^{n/2} p_i^{(h_j)}\right)\right]},
$$
where the approximation in the above statement is of the order as in statement~\ref{eqn:6}, enhanced by $\mp(n/2)^2 T'\epsilon'$, following statement~\ref{eqn:7}. Then
\begin{equation}
\displaystyle{\sum\limits_{i=1}^{N}\sum\limits_{k=1; k\neq i}^{N} X_i^{(g)}X_k^{(g)}} \approx 
\displaystyle{\sum\limits_{i=1}^{N}\sum\limits_{k=1; k\neq i}^{N} X_i^{(h)}X_k^{(h)}}, 
\label{eqn:8}
\end{equation}
where the approximation is of the order of $(n/2)^2 T'\epsilon'$. 

Since $\displaystyle{\left(\sum\limits_{i=1}^{N} X_i^{(g)}\right)^2 }- \displaystyle{\sum\limits_{i=1}^{N}\left(X_i^{(g)}\right)^2}=\displaystyle{\sum\limits_{i=1}^{N}\sum\limits_{k=1; k\neq i}^{N} X_i^{(g)}X_k^{(g)}}$, as in the last line of Equations~\ref{eqn:1}--statement~\ref{eqn:8} tells us that for the $g$-th and the $h$-th sub-tests,
\begin{equation}
\displaystyle{\left(\sum\limits_{i=1}^{N} X_i^{(g)}\right)^2 }- \displaystyle{\sum\limits_{i=1}^{N}\left(X_i^{(g)}\right)^2} \approx
\displaystyle{\left(\sum\limits_{i=1}^{N} X_i^{(h)}\right)^2 }- \displaystyle{\sum\limits_{i=1}^{N}\left(X_i^{(h)}\right)^2},
\label{eqn:lastbut1}
\end{equation}
where, as for statement~\ref{eqn:8}, the approximation is of the order of $(n/2)^2 T'\epsilon'$. 
But by squaring both sides of Equation~\ref{eqn:basic} we get $\displaystyle{\left(\sum\limits_{i=1}^{N} X_i^{(g)}\right)^2}\approx \displaystyle{\left(\sum\limits_{i=1}^{N} X_i^{(g)}\right)^2}$, where the approximation is of the order of $\epsilon^2\mp 2 T\epsilon$. Then in statement~\ref{eqn:lastbut1} we get
\begin{equation}
\displaystyle{\sum\limits_{i=1}^{N}\left(X_i^{(g)}\right)^2} \approx
\displaystyle{\sum\limits_{i=1}^{N}\left(X_i^{(h)}\right)^2},
\end{equation}
the approximation in which is of the order of $\epsilon^2\mp 2T\epsilon \pm (n/2)^2 T'\epsilon'$, i.e. of the order of $\epsilon^2$.
\end{proof}

\begin{remark}
{It is to be noted in the proof of Theorem~\ref{theorem:vars} that even if the sum of scores of the two sub-tests are equal, as per our method of splitting, i.e. even if $\epsilon$ is achieved to be 0, absolute difference between sums of squares of scores in the two sub-tests is not necessarily 0, since $T$ and $(n/2)^2 T'$ are not necessarily equal.}
\end{remark}

Now, variance of sub-tests $g$ and $h$ is respectively, 
$\displaystyle{\sum\limits_{i=1}^{N}
  \frac{\left(X_i^{(g)}\right)^2}{N} - {\bar{X}}_g^2}$ and
$\displaystyle{\sum\limits_{i=1}^{N}
  \frac{\left(X_i^{(h)}\right)^2}{N} - {\bar{X}}_h^2}$. Then near-equality of sums of squares of sub-test scores, (Theorem~\ref{theorem:vars}) and near-equality of means (Theorem~\ref{theorem:means}) imply that the difference between sub-test variances is small.
Empirical confirmation of
near-equality of sub-test variances is presented in later sections.
Now item scores manifest item difficulty. Therefore, splitting using
item scores is equivalent to splitting using item difficulty values.
From the near-equality of variances it follows that, if instead of
allocating items into the two sub-tests, on the basis of the
difficulty value of the item, we split the test according to item
variance, we would get nearly the same sub-tests. Thus our method of
dichotomisation with respect to item scores (or difficulty values), is
nearly equivalent to dichotomisation with respect to item variance.

The near-equality of means and variances of the sub-tests are
indicative of the sub-tests being nearly parallel, since parallel
tests have equal means and variances. Then the item scores of the two
sub-tests can be taken as coming from nearly the same populations with
nearly same density functions having two parameters, namely mean and
variance.

For parallel sub-tests, variances are equal, and the regression of the
item score vector of the $g$-th sub-test on the $h$-th, coincides with
the regression of the scores of the $h$-th sub-test on that of the
$g$-th; such coinciding regression lines imply that the Pearson
product moment correlation $r_{gh}$ between $\{X_i^{(g)}\}_{i=1}^N$
and $\{X_i^{(h)}\}_{i=1}^N$, or the split-half regression
coefficient $r_{gh}$, is maximal, i.e. unity \cite{stepniak_2009}. Our
method of splitting the test results in nearly parallel sub-tests, so
that the split-half regression is close to, but less than unity. The
closer $\vert{\cal S}\vert$ is to 0, the more parallel are the
resulting sub-tests $g$ and $h$, and higher is the value of the
attained $r_{gh}$. Thus, the method gives a simple way of splitting a
test in halves, while ensuring that the split-half reliability
$r_{tt}$ is maximum.

The iterative method also ensures that the two sub-tests are almost
equi-correlated with a third variable. Thus, these (near-parallel)
sub-tests will have almost equal validity.  

The problem of splitting an even number of non-negative item scores
into two sub-tests with the same number of items, such that absolute
difference between sums of sub-test item scores is minimised, is a
simpler example of the partition problem that has been addressed in
the literature \cite[among others]{hayes,borgs_2001}. Detailed
discussion of these methods is outside the scope of this paper, but
within the context of test dichotomisation aimed at computing
reliability, we can see that the method of assigning even numbered
items to one sub-test and odd-numbered ones to another as a method of
splitting a test into two halves, cannot yield sub-tests that are as
parallel as the sub-tests achieved via our method of splitting,
because an odd-even assignment does not guarantee that the sum of item
scores in one sub-test approaches that of the other closely, unless
the difficulty value of all items are nearly equal. Thus, a test in
which consecutive items are of widely different difficulty values,
will yeld sub-tests with means that are far from each other, if the
test is dichotomised using an odd-even method of assignment of items
to the sub-tests; in comparison, our method of splitting is designed
to yield sub-tests with close values of variances and even closer
means. Following on from the low order of our algorithm (see
Remark~\ref{remark:remark1}), a salient feature of our method is that
it is a fast method, the order of which depends on the number of items
in the test and not the examinee number, thus allowing for fast
splitting of the test score data and consequently, fast computation of
reliability.


\subsection{Estimation of True Score}
\label{sec:true_score}
\noindent

When the reliability is computed according to the
method laid out above, it is possible to perform interval estimation
of true score of an individual whose obtained score is known, where
the interval itself represents the ±1 standard deviation of the
distribution of the error score $E$. In other words, the
distribution of the error score $E$ is assumed--normal with zero mean
and variance $S_E^2 = S_X^2 (1-r_{tt})$ (as per the classical
definition of reliability); using the error variance, errors on the
estimated true scores are given as ±$S_E$. Below we discuss this
choice of model for the interval on the estimate of the true scores.

At a given observed score $X$, the true score can be estimated using the linear regression:
\begin{equation}
\hat{T}=\alpha_1 + \beta_1 X 
\label{eqn:te}
\end{equation}
where the regression coefficients are
$\beta_1=r_{XT}\displaystyle{\frac{S_T}{S_X}}=r_{tt}$ and
$\alpha_1={\bar{X}}(1-\beta_1)$. 

That the error on the estimated true score is given by $\pm S_E$, is
indeed a model choice \cite{klinebook} and is in fact, an
over-estimation given that the error $\varepsilon_i$ of estimation is
$$\varepsilon_i^2=(\alpha_1 + \beta_1 X_i - X_i)^2 =
(\bar{X}(1-r_{tt}) + r_{tt} X_i - X_i)^2=
[(X_i-\bar{X})(1-r_{tt})]^2$$ (recalling that
$\alpha_1=\bar{X}(1-r_{tt})$ and $\beta_1=r_{tt}$). Then
$$\displaystyle{\frac{\sum\limits_{i=1}^{N} \varepsilon_i^2}{N}}=(1 -
r_{tt})^2 S_X^2 \leq (1 - r_{tt}) S_X^2 = S_E^2,$$ i.e.  $S_E^2 -
\displaystyle{\frac{\sum\limits_{i=1}^{N} \varepsilon_i^2}{N}}= (1 -
r_{tt}) S_X^2 - (1 - r_{tt})^2 S_X^2=r_{tt} (1 - r_{tt}) S_X^2$. In
other words, the standard error of prediction of the true score from a
linear regression model, falls short of the error variance $S_E^2$ by
the amount $r_{tt}(1- r_{tt})S_X^2$. Thus, the higher the reliability
$r_{tt}$ of the test, the less is the difference between the standard
error of prediction and our model of the error (on ${\hat{T}}$). In
general, our model choice of the error on ${\hat{T}}$ is higher than the standard error of prediction, for a given $X$. 


%


Using ${\hat{T}}$ of each individual taking the test, one may
undertake computation of the probability that the percentile true
score of the $i$-th examinee is $t$ given the observed percentile
score of the examinee is $x$ and reliability is $r$, i.e. $\Pr(T \leq
t\vert r_{tt}=r, X\leq x)$.  We illustrate this in our analysis of
real and simulated test score data.


In the context of estimating true scores using a computed reliability, we realise that using the split-half reliability $r_{gh}$ to estimate true scores ${\widehat{T}}_{split-half}$ will be in excess of the estimate ${\widehat{T}}_{classical}$ obtained by using the classically defined reliability $r_{tt}$, for high values of the observed scores, and under-estimated compared to ${\widehat{T}}_{classical}$ for low values of $X$. This can be easily realised.
\begin{theorem}
\label{theorem:over}
${\widehat{T}}_{split-half} \geq {\widehat{T}}_{classical}$ for $X > {\bar{X}}$ and ${\widehat{T}}_{split-half} \leq {\widehat{T}}_{classical}$ for $X < {\bar{X}}$.
\end{theorem}
\begin{proof}
\begin{eqnarray}
{\widehat{T}}_{split-half} - {\widehat{T}}_{classical} &=& [{\bar{X}} + r_{gh}(X-{\bar{X}})] - [{\bar{X}} + r_{tt}(X-{\bar{X}})]\\ \nonumber
&=&  (X - {\bar{X}})(r_{gh} - r_{tt}) \\ \nonumber
\geq 0 \quad {\mbox{for}}\quad X > {\bar{X}} &{\mbox{and}}& \leq 0 \quad {\mbox{for}} \quad X < {\bar{X}}  
\end{eqnarray}
given that $\quad r_{gh} \geq r_{tt}$.\\
Therefore ${\widehat{T}}_{split-half} \geq {\widehat{T}}_{classical}$ for $X \geq {\bar{X}}$ and ${\widehat{T}}_{split-half} \leq {\widehat{T}}_{classical}$ for $X \leq {\bar{X}}$.
\end{proof}

\section{Reliability of a battery}
\noindent
The above method of finding reliability, as per the classical
definition, can be extended to compute the reliability of a battery of
tests. Battery scores of a set of examinees are obtained as a function
of the scores of the constituent tests. After administration of the
battery to $N$ individuals, $S_X^2$, $S_T^2$, $S_E^2$ and $r_{tt}$ of
each constituent test are known as per the method described above. The
method of obtaining the reliability of the battery will depend on
these parameters as well as on the definition of the battery
scores. Usually, battery scores are computed as the sum of the scores
in the individual constituent tests (summative scores) or as a
weighted sum of these test scores. However, the possibility of a
non-linear combination of test scores to depict battery scores cannot
be ruled out. Here we discuss the weighted sum of component test scores
after motivating the concept of summative scores; we illustrate our
method of finding the weights on real data
(Section~\ref{sec:real}). 

\subsection{Weighted sum of component tests}
\label{sec:weights}
\noindent
Suppose a battery consists of two tests: Test-1 and Test-2. Let the
$i$-th examinee's scores in the 2 constituent tests be $X_{1i}$ and
$X_{2i}$. Let this examinee's true scores and error score of the
$j$-th constituent test be $T_{ji}$ and $E_{ji}$ respectively,
$j=1,2$. Then the summative score of the $i$-th examinee is $X_i =
X_{1i} + X_{2i}$ assuming additivity and equal weights. Let
$r_{tt(1)}$ and $r_{tt(2)}$ denote reliability of the respective
constituent tests.

Now,
\begin{equation}
{\textrm var}(X)  =  {\textrm var}(X_1)  +  {\textrm var}(X_2)  +  2{\textrm cov}(X_1, X_2)
\label{eqn:varcov}
\end{equation} 
suggests
\begin{eqnarray}
\displaystyle{\sum X_{1i}X_{2i}}    &=&  \displaystyle{\sum(T_{1i} + E_{1i}) (T_{2i} + E_{2i})} \nonumber \\
& =& \displaystyle{\sum T_{1i}T_{2i}} \nonumber \\
\end{eqnarray}
since $\displaystyle{\sum T_{1i}E_{2i}}=\displaystyle{\sum T_{2i}E_{1i}}=\displaystyle{\sum E_{1i}E_{2i}}=0$.

Again
\begin{eqnarray}
{\textrm cov}(X_1, X_2) &=& \displaystyle{\frac{\sum X_{1i}X_{2i}}{N} - {\bar X}_1 {\bar X}_2 }\nonumber \\  
&=& \displaystyle{\frac{\sum T_{1i}T_{2i}}{N} - {\bar T}_1 {\bar T}_2 }\nonumber \\  
&=&{\textrm cov}(T_1, T_2)  
\label{eqn:t1t2}
\end{eqnarray}
Now  variance of true score of the battery is    
\begin{eqnarray}
S_{T(Battery)} &=& {\textrm var}(T_1) + {\textrm var}(T_2) + 2{\textrm cov}(T_1, T_2)\nonumber \\  
&=& {\textrm var}(T_1) + {\textrm var}(T_2) + 2{\textrm cov}(X_1, X_2)\nonumber \\  
&=& r_{tt(1)}S_{X_1} + r_{tt(2)}S_{X_2}  + 2{\textrm cov}(X_1, X_2)
\label{eqn:rel_bat}
\end{eqnarray}
where we have used Equation~\ref{eqn:t1t2}. Thus, reliability of a battery can be found using Equation~\ref{eqn:rel_bat}. Using these equations, 
reliability of the battery $r_{tt(battery)}$  is given by   
\begin{equation}
r_{tt(battery)} = \displaystyle{\frac{r_{tt(1)}S_{X_1} + r_{tt(2)}S_{X_2}  + 2{\textrm cov}(X_1, X_2)}{S_{X_1} + S_{X_2}  + 2{\textrm cov}(X_1, X_2)}}
\end{equation}
In general, if a battery has $K$-tests and all tests are equally important then the reliability of the battery, where the battery score is a summative score of the $K$ constituent tests, will be given by
\begin{equation}
r_{tt(battery)} = \displaystyle{\frac{\sum_{i=1}^K r_{tt(i)}S_{X_i} + 
\sum_{i=1, i\neq j}^K \sum_{j=1}^K 2{\textrm cov}(X_i, X_j)}{\sum_{i=1}^K S_{X_i} + \sum_{i=1, i\neq j}^K \sum_{j=1}^K 2{\textrm cov}(X_i, X_j)}}
\label{eqn:rel_bat_gen}
\end{equation}                             
We can extend this idea to weighted summative scores.

Consider a battery of $K$-tests with weights
$W_1, W_2, \ldots, W_k$ where $W_i > 0\forall\:i= 1,2,\ldots,k$ and
$\displaystyle{\sum_{i=1}^K W_i} = 1$.  To find reliability of such a
battery, let us consider the composite score or the battery score $Y=
\displaystyle{\sum_{i=1}^K W_i X_i}$ where $X_i$ is the score of the
  $i$-th constituent test. Clearly, ${\textrm{var}}(Y) =
  \displaystyle{\sum_{i=1}^K W_i^2{\textrm{var}}(X_i)}$.

Then
\begin{equation}
r_{tt(battery)} = \displaystyle{\frac{\sum_{i=1}^K r_{tt(i)}W_i^2 S_{X_i} + 
\sum_{i=1, i\neq j}^K \sum_{j=1}^K 2 W_i W_j {\textrm cov}(X_i, X_j)}{\sum_{i=1}^K W_i^2 S_{X_i} + \sum_{i=1, i\neq j}^K \sum_{j=1}^K 2W_i W_j {\textrm cov}(X_i, X_j)}}
\label{eqn:rel_bat_wt}
\end{equation}                             

However, a major problem in the computation of the reliability
coefficient in this fashion could be experienced in determination of the
weights. 

\subsection{Our method for determination of weights} 
\label{sec:wtdetermination}
\noindent
We propose a method towards the determination of the vector of
weights $\bW = (W_1,W_2,\ldots,W_K)^T$ with
$\displaystyle{\sum_{i=1}^K W_i} = 1$, such that the variance of the
battery score is a minimum where the battery score vector is $\bY$
with  ${\textrm{var}}(\bY) = \bW^T \bD \bW$ and $\bD$ is the
variance-covariance matrix of the component tests.  From a single
administration of a battery to $N$ examinees, the variance-covariance
matrix $\bD$ is such that the $i$-th diagonal element of this matrix
is the variance $S_{X_i}$ of the $i$-th constituent test, and the
$ij$-th off diagonal element is ${\textrm{cov}}(X_i,X_j)$.  This is
subject to the condition $\bW^T\bee=1$ where $\bee$ is the
$K$-dimensional vector, the $i$-th component of which is 1,
$\forall\:i=1,\ldots,K$. It is possible to determine the weights using
the method of Lagrangian multipliers $\lambda$.  To this effect, we
define $A := \bW^T \bD \bW + \lambda(1 - \bW^T\bee)$ and set the
derivative of $A$ taken with respect to $\bW$ and $\lambda$ to zero
each, to respectively attain:
\begin{equation}
2\bD\bW - \lambda\bee= 0 \quad\mbox{and}\quad 1 - \bW^T\bee = 0.\nonumber 
\end{equation}
Thus,
\begin{eqnarray}
\bW = \displaystyle{\frac{\lambda}{2}{\bD}^{-1}\bee}\quad&\mbox{and}&\quad 
\bW^T\bee = 1,\nonumber \\
\mbox{or}\quad \bW= \displaystyle{
\frac{\bD^{-1}\bee}{\bee^T\bD^{-1}\bee}
} \quad&\mbox{and}&\quad \lambda= \displaystyle{\frac{2}{\bee^T\bD^{-1}\bee}} 
\label{eqn:de}
\end{eqnarray}

\subsection{Benefits of this method of finding weights}
\noindent
This method of finding reliability of a battery of a number of tests
can be used in the context of summation score, difference score or
weighted sum of scores. Weights found as above have the advantage that
the battery score or composite score ($\bY$) has minimum variance.
Also, covariance between the battery score and the test score of an
individual test is a constant, i.e. ${\textrm{cov}}(Y_i,X_i)=
\displaystyle{\frac{1}{\bee^T\bD^{-1}\bee}}$ $\forall\:i$.  If the
available test scores are standardised and independent such that the
$i$-th score is $Z_i$, then weights are equal, and correlation between
$Y$ and $Z_i$ is the same as correlation between $Y$ and $Z_j$ =
$\displaystyle{\frac{1}{\sqrt{\bee^T\bR^{-1}\bee}}}$ $\forall\:i=j$,
$i,j=1,2,\ldots,k$, where $\bR$ is the correlation matrix.  In other
words, the battery score is equi-correlated with the standardised
score of each constituent test.  The method alone does not guarantee
that all weights be non-negative. This is especially true if the tests
are highly independent so that $\bR$ is too sparse. Non-negative weights
can be ensured by adding the physically motivated constraint that $W_i
\geq 0$ $\forall\: i$, in which case, the problem boils down to a
Quadratic Programming formulation. This is similar to the data-driven
weight determination method presented by Chakrabartty (2013).

If it is found desirable to get weights so that $X_i$ is
  proportional to ${\textrm{cov}}(Y, X_i)$ then it can be proved that
  $\bW$ is the eigenvector corresponding to the maximum eigenvalue of
  the variance-covariance matrix of the test scores.  
If on the other hand, we want $W_i$ as proportional to
  $\displaystyle{\frac{{\textrm{cov}}(Y,X_i)}{{\textrm{var}}(X_i)}}$,
  then $\bW = \displaystyle{\frac{\bS^{-1} U}{\bee^T \bS^{-1} U}}$
  where $\bS$ is the diagonal matrix of the standard deviations of the
  constituent tests and $U$ is the maximum eigen-value of the
  correlation matrix.

In order to find the battery scores, weights of test scores can be
found by principal component analysis, factor analysis, canonical
reliability, etc. However, it is suggested that reliability of the
composite scores be found as weighted scores
(Equation~\ref{eqn:rel_bat_wt}) where the weights are determined as per
Equation~\ref{eqn:de}.

So far, we have discussed additive models. We could also have
 multiplicative models like $Y= \displaystyle{\prod X_i^{W_i}}$ so that
 on a log-scale, we retrieve the additive model.

We illustrate our method of finding weights using real data in
Section~\ref{sec:real}.

\section{Application to simulated data}
\label{sec:simulated}
\noindent
In order to validate our method of computing the classically defined
reliability following dichotomisation of a test into parallel groups,
we use our method to find the value of $r_{tt}$ of 4 toy tests, the
scores of which are simulated from chosen models, as described
below. We simulate the 4 test data sets $D_1, D_2, D_3, D_4$ under 4
distinct model choices; the underlying standard model is that the score
$X_{i}^{(j)}$ is obtained by the $i$-th examinee to the $j$-th item in a test,
is a Bernoulli variate, i.e.
$$X_{i}^{(j)}\sim{\mbox{Bernoulli}}(p)\quad {\mbox{implying}}\quad\Pr(X_{i}^{(j)}=1)=p,\:\:\Pr(X_{i}^{(j)}=0)=1-p$$ 
where the probability $Pr(X_{i}^{(j)}=1)$, of answering the $j$-th item correctly by the $i$-th examinee
\begin{itemize}
\item is held a constant $p_i$ for the $i$-th examinee $\forall\:j=1,\ldots,n$, with $p_i$ sampled from a uniform distribution in [0,1], $\forall\:i=1,\ldots,N$, in data $D_1$. 
\item is held a constant $p_i$ for the $i$-th examinee $\forall\:j=1,\ldots,n$, with $p_i$ sampled from a Normal distribution ${\cal N}(0.5,0.2)$, $\forall\:i=1,\ldots,N$, in data $D_3$. 
\item is held a constant $p_j$ for the $j$-th item for all examinees, with $p_j$ sampled from a uniform distribution in [0,1], $\forall\:j=1,\ldots,n$, in data $D_2$. 
\item is held a constant $p_j$ for the $j$-th item for all examinees, with $p_j$ sampled from a Normal distribution ${\cal N}(0.5,0.2)$, $\forall\:j=1,\ldots,n$, in data $D_4$. 
\end{itemize}
We use $n$=50 and $N$=999 in our simulations.

Thus we realise that data sets $D_1$ and $D_3$ resemble test data in
reality, with the ability of the $i$-th examinee represented by
$p_i$. Our simulation models are restrictive though in the sense that
variation with items is ignored. Such a test, if administered, will
not be expected to have a low reliability. On the other hand, the data
sets $D_2$ and $D_4$ are toy data sets that are utterly unlike real
test data, in which the probability of correct response to a given
item is a constant, irrespective of examinee ability, and examinee
ability varies with item--equally for all examinees. A toy test data
generated under such an unrealistic model would manifest low test
variance and therefore low reliability. Given this background, we
proceed to analyse these simulated tests with our method of computing
$r_{tt}$.

We implement our method to compute reliabilities of all 4 data sets (using Equation~\ref{eqn:4thstep}). The results are given in Table~1.
\begin{table}
\renewcommand\baselinestretch{1.5}{
\caption{Table showing results of using our method of dichotomisation
  of 4 simulated test data sets, $D_1,\ldots,D_4$, into 2 parallel
  sub-tests, $g$ and $h$ , resulting in the computation of
  classically-defined reliability $r_{tt}$ of the test (see
  Section~\ref{sec:dichot} for details of our method).}}
 \label{table:table1}
\hspace*{-1cm}
{\footnotesize
\begin{tabular}{|c||c||c|c||c|c||c||c|}\hline
\hline
\multicolumn{1}{|c||}{Data set}& \multicolumn{1}{c||}{Test}& \multicolumn{2}{|c||}{Sum of item scores $X_g \& X_h$ in sub-tests} & \multicolumn{2}{|c||}{Sum of squares of item scores in sub-test} & \multicolumn{1}{c||}{$\sum_{i=1}^{N}X^{(g)}_iX^{(h)}_i$}& \multicolumn{1}{|c|}{Reliability}\\
& variance & $g$  & $h$ & $g$  & $h$ & & $r_{tt}$\\
\hline
${D_1}$ & 222.89 & 12014$\quad\quad\quad\quad$ & 12013 & 202506$\quad\quad\quad\quad$  & 201523 & 198257 & 0.9662 \\
${D_2}$ & 7.96 & 10440$\quad\quad\quad\quad$  & 10439 & 113212$\quad\quad\quad\quad$  & 112843 & 109133 & 0.02050 \\
${D_3}$ & 110.06 & 12597$\quad\quad\quad\quad$  & 12597 & 188727$\quad\quad\quad\quad$  & 189465 & 183564 & 0.8994 \\
${D_4}$ & 10.86 & 12683$\quad\quad\quad\quad$  & 12684 & 166199$\quad\quad\quad\quad$  & 166520 & 161130 & 0.0361\\
\hline
\end{tabular}
}
\vspace{-0.1in}
\end{table}

\begin{figure}[!h]
\centering
\vspace*{1cm}
  \includegraphics[width=18cm]{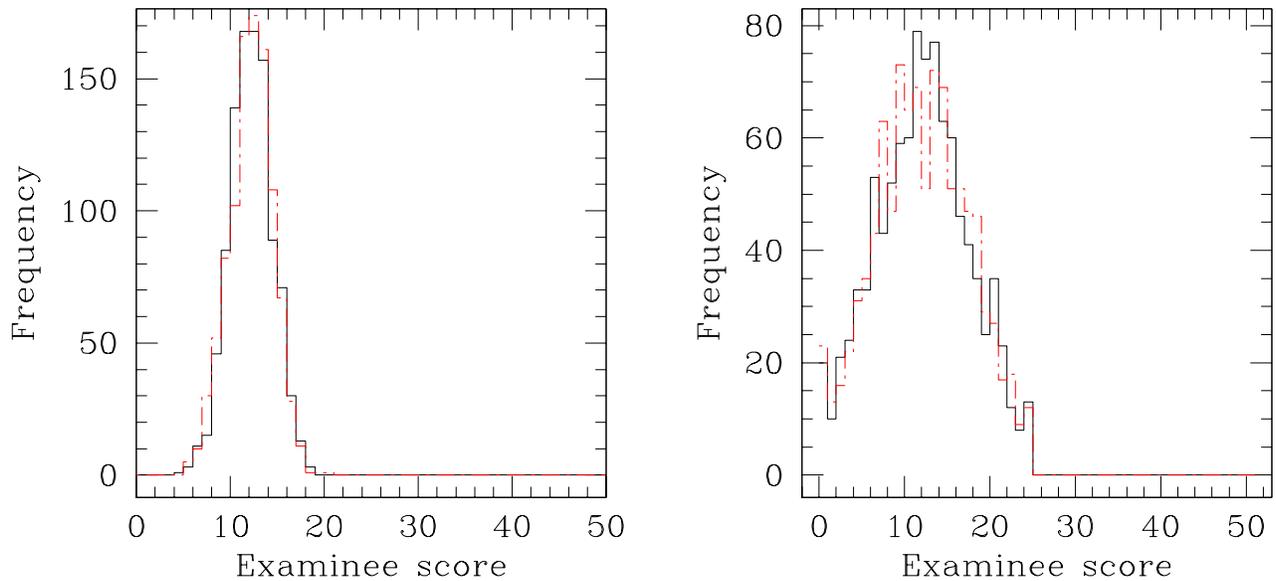}
\vspace*{-1cm}
\caption{\small{Figure showing 
    histograms of the scores of 999
    examinees in the 2 sub-tests (in solid and broken lines
    respectively), obtained by splitting the simulated test data set
    $D_3$ which has been generated under the choice that examinee
    ability is normally distributed (right panel). The left panel
    includes histograms of the 2 sub-tests that result from splitting
    the test data $D_4$ that was simulated using examinee
    ability as item-dependent, with probability for correct response
    to the $j$-th item given as a normal variate, unrealistically
    fixed for all examinees, $\forall\:j=1,\ldots, 50$.}}
\label{fig:new}
\end{figure}

The reliabilities of tests with data $D_1$ and $D_3$ are as expected
high, while the same for the infeasible tests $D_2$ and $D_4$ are low,
as per prediction. We take these results as a form of validation our method
of computing $r_{tt}$ based on the splitting of a test. 

In Figure~\ref{fig:new}, we present the histograms of the 2 sub-tests
that result from splitting the realistic data $D_3$ as well histograms
of the 2 sub-tests that result from the splitting of the unrealistic
test data $D_4$.

%


\subsection{Results from Large Simulated Tests}
\noindent
We also conducted a number of experiments with finding reliabilities of larger test data sets that were simulated. The simulations were performed such that the test score variable has a Bernoulli distribution with parameter $p$. in these simulations, we chose $p_i$ as fixed for the $i$-th examinee, with $p_i$ randomly sampled from a chosen Gaussian $pdf$, i.e. $p_i\sim{\cal N}(0.5, 0.2)$ by choice, $i=1,\ldots, N$. We simulated different test score data sets in this way, including
\begin{enumerate}
\item[--] a test data set for 5$\times$10$^5$ examinees taking a 50-item test.
\item[--] a test data set for 5$\times$10$^4$ examinees taking a 50-item test.
\item[--] a test data set for 1000 examinees taking a 100-item test.
\item[--] a test data set for 1000 examinees taking a 1000-item test.
\end{enumerate}
In each case, the test data was split using our method and reliability of the test was computed as per the classical definition. The 4 simulated test data sets mentioned above, yielded reliabilities of 0.96, 0.98, 0.93, 0.85, in order of the above enumeration. Histograms of the sub-tests obtained by splitting each test data were over-plotted to confirm their concurrence. 

Importantly, the run-time of reliability computation of these large cohorts of examinees ($N$=500,000 and 50,000), who take the 50-item long test, is very short--from about 0.8 seconds for the 50,000 cohort to about 6.2 seconds for the 500,000 cohort. On the other hand the order of our splitting algorithm being ${\cal O}((n/2)^2)$, the run-times increase rapidly for the 1000-item test, from the 100-item one, with the fixed examinee number. These experiments indicate that the computation of reliabilities for very large cohorts of examinees, in a test with a realistic number of items, is rendered very fast indeed, using our method.

\section{Application to real data}
\label{sec:real}
\noindent
We apply our method of computing reliability using the classical
definition, following the splitting of a real test into 2 parallel
halves by the method discussed in Section~\ref{sec:dichot}. We also
estimate the true scores and offer the error in this estimation, using
the computed reliability and a real test score data set. The used real
test data was obtained by examining 912 examinees in a multiple choice
examination administered with the aim of achieving selection to a
position. The test had 50 items and maximum time given was 90
minutes. To distinguish between this real test data and other real
test data sets that we will employ to determine reliability of a
battery, we refer to the data set used in this section as DATA-I. For
the real test data DATA-I, the histograms of the item scores and that
of the examinee scores are shown in Figure~\ref{figure:histos}. The
histogram of the examinee scores shows that $X_i
\leq 26\:\forall i=1,\ldots,912$, with the biggest mode of the score
distribution around 11. In other words, the test was a low-scoring
one. DATA-I indicates 2 other smaller modes at about 15
and 22. Some other parameters of the test data DATA-I are as listed below.
\begin{enumerate}
\item[-] number of items $n$=50,
\item[-] number of examinees is $N$=912,
\item[-] magnitude of the vector of maximum possible scores is 
$\parallel \bI \parallel$ = $\sqrt{912\times 50^2}\approx 1509.97$,
\item[-] magnitude of the observed score vector is 
$\parallel \bX \parallel\approx$ 357.82,
\item[-] $\parallel \bX \parallel^2$ = 128032,
\item[-] $\cos\theta_X\approx$0.9275
\item[-] test mean ${\bar X}\approx$10.99, 
\item[-] test variance $S_X^2 = \approx 19.63$.
\end{enumerate}      

\begin{figure}[!t]
\centering
\vspace*{1cm}
  \includegraphics[width=18cm]{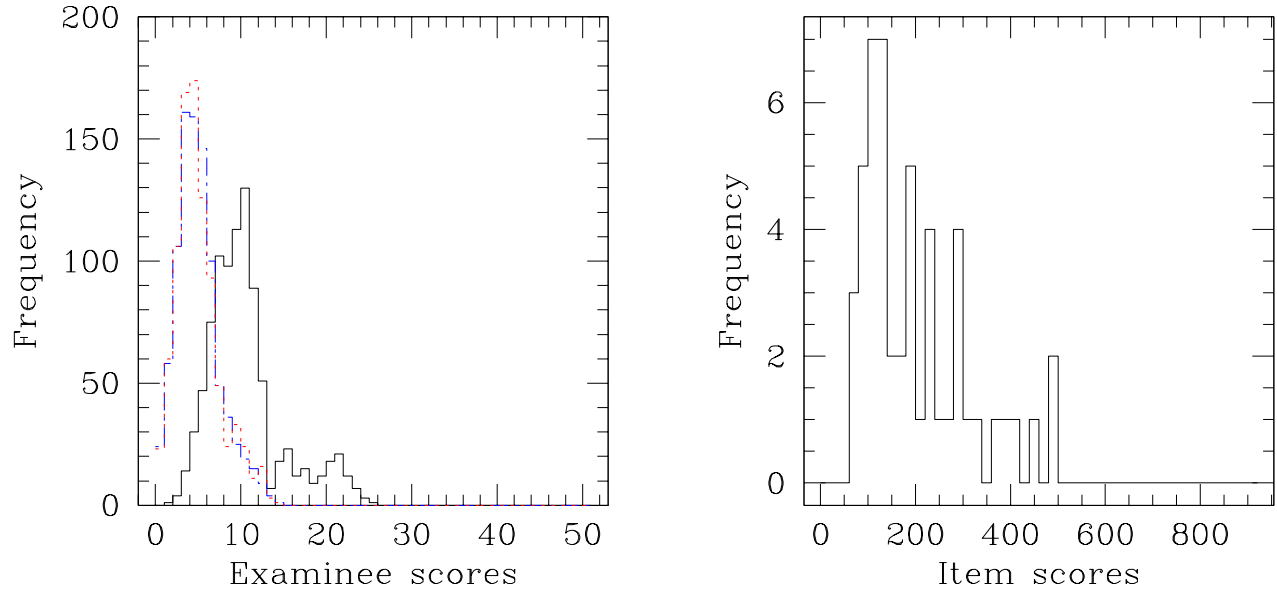}
\vspace*{-1cm}
\caption{\small{Figure showing histograms of the 912 examinee scores (left panel) and of 50 item scores (right panel) of the real test data DATA-I that we use to implement our method of splitting the test, aimed at computation of reliability as per the classical definition. The left panel also includes histograms of the $g$ and $h$ sub-tests that result from the splitting of this test; the histograms of the sub-tests are shown in dotted and dashed-dotted lines, superimposed on the histogram of examinee scores for the full test (in solid black line).}}
\label{figure:histos}
\end{figure}

\begin{figure}[!t]
\centering
  \vspace*{-4cm}
\includegraphics[width=18cm]{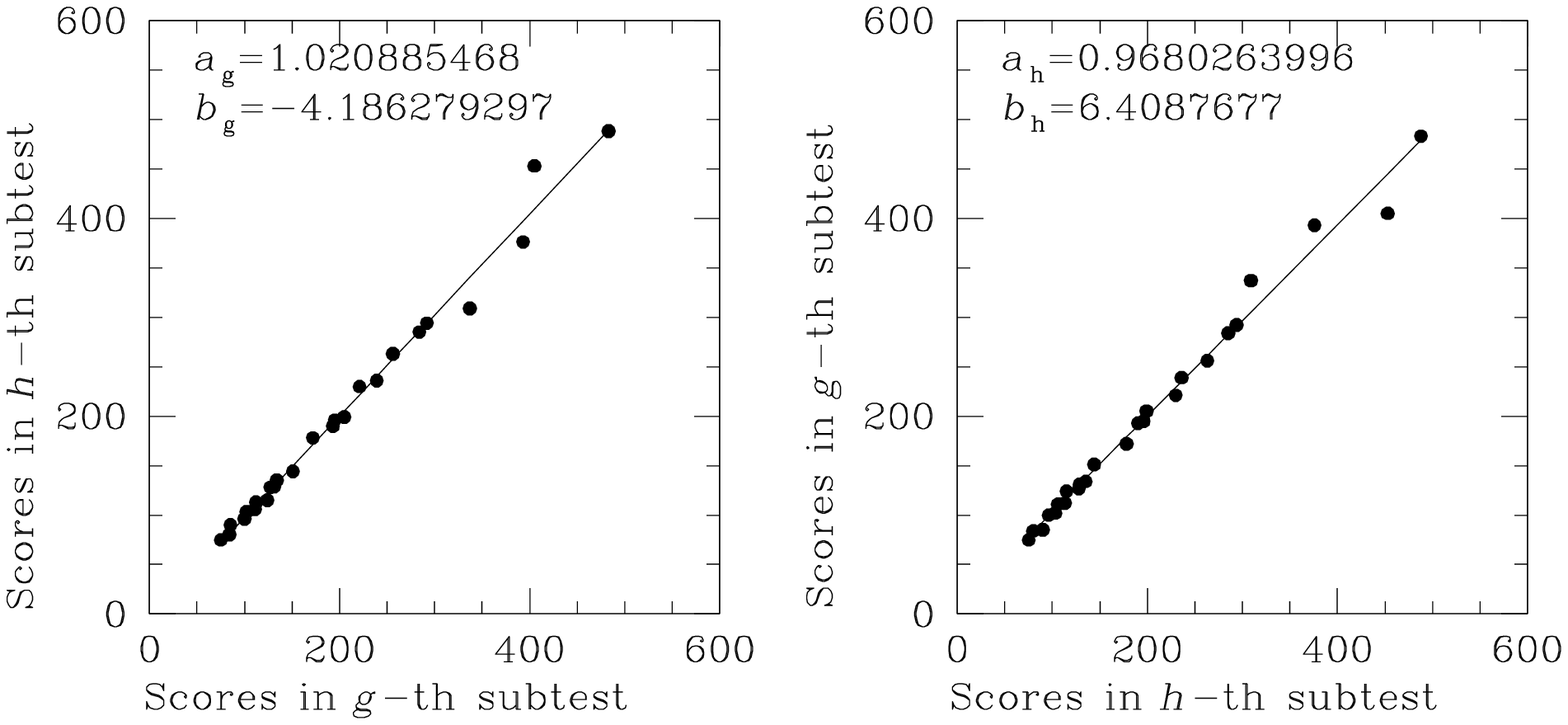}
\vspace*{-14cm}
\caption{\small{Figure showing regression of the examinee scores in the sub-test $g$ on the scores in sub-test $h$ (right panel) and vice versa (left panel), where the 2 sub-tests are obtained by splitting real data DATA-I with our method of dichotomisation (Section~\ref{sec:dichot}). The score data are shown in filled circles while the regression lines are also drawn. The regression coefficients are marked on the top left hand corner of the respective panel.} }
\label{figure:regress}
\end{figure}

The test was dichotomised into parallel halves by the iterative
process discussed in Section~\ref{sec:dichot}. Let the resulting
sub-tests be referred to as the $g$ and $h$ sub-tests. Then each of
these sub-tests had 25 items in it and sum of examinee scores is
$\displaystyle{\sum\limits_{i=1}^{912} X_i^{(g)}}=
\displaystyle{\sum\limits_{i=1}^{912}
  X_i^{(h)}}=5011$, so that the mean of each of the two sub-tests is equal to about 5.49. Also, variances of the 2 sub-tests are approximately equal at 6.81 and 6.49. 
The histograms of the examinee scores in these 2 sub-tests are drawn
in blue and red in the left panel of Figure~\ref{figure:histos}. That the
histograms for the $g$ and $h$ sub-tests overlap very well, is
indicative of these sub-tests being strongly parallel. We regress the
score vector $(X_1^{(g)},\ldots,X_{912}^{(g)})^T$ on the score vector
$(X_1^{(h)},\ldots,X_{912}^{(h)})^T$; the regression lines are shown
in the right panel of Figure~\ref{figure:regress}. Similarly, regressing
$(X_1^{(h)},\ldots,X_{912}^{(h)})^T$ on the score vector
$(X_1^{(g)},\ldots,X_{912}^{(g)})^T$ results in a similar linear
regression line (shown in the left panel of
Figure~\ref{figure:regress}). Table~2 gives the details of splitting of
the 50 items of the full test into the 2 sub-tests.

\begin{table}
\caption{Splitting of real data DATA-I of a selection test with 50 items, administered to 912 examinees. The dichotomisation has been carried out according to  the algorithm discussed in Section~\ref{sec:dichot}.}
 \label{table:table1}
\begin{center}
{\footnotesize
\begin{tabular}{|c|c||c|c||c|}\hline
\multicolumn{2}{|c||}{$g$-th sub-test} & \multicolumn{2}{c||}{$h$-th sub-test} & \multicolumn{1}{c|}{Difference between scores of 2 tests}\\
\hline
Score & Item No. & Score & Item no. & \\
\hline
75 & 25 & 75 & 1 & 0\\
84 & 39 & 80 & 46 & 4\\
85 & 43 & 90 & 24 & -5\\
100 & 20 & 96 & 9 & 4\\
102 & 44 & 103 & 34 & -1\\
111 & 50 & 106 & 31 & 5\\
112 & 5 & 113 & 41 & -1\\
124 & 32 & 115 & 45 & 9\\
127 & 36 & 128 & 8 & -1\\
131 & 33 & 129 & 26 & 2\\
134 & 18 & 135 & 3 & -1\\
151 & 29 & 144 & 4 & 7\\
172 & 7 & 178 & 48 & -6\\
193 & 19 & 190 & 21 & 3\\
195 & 37 & 196 & 14 & -1\\
205 & 28 & 199 & 47 & 6\\
221 & 30 & 230 & 6 & -9\\
239 & 27 & 236 & 49 & 3\\
256 & 11 & 263 & 12 & -7\\
284 & 2 & 285 & 17 & -1\\
292 & 22 & 294 & 10 & -2\\
337 & 15 & 309 & 23 & 28\\
393 & 13 & 376 & 42 & 17\\
405 & 16 & 453 & 38 & -48\\
483 & 35 & 488 & 40 & -5\\
\hline
Sum of item scores & 5011 & & 5011 & \\
Sum of squares of item scores & 33453 & & 33747 & \\
\hline
\end{tabular}
}
\vspace{-0.1in}
\end{center}
\end{table}

\subsection{Computation of reliability}
\noindent
We implement the splitting of the real test score data into the scores
in the $g$ and $h$ sub-tests to compute the reliability as per the
theoretical definition, i.e. as given in
Equation~\ref{eqn:4thstep}. Then using the observed sub-test score
vectors $\bX^{(g)}$ and $\bX^{(h)}$ we get the classically defined reliability to be $r_{tt} \approx 0.66$. Then
$r_{XT}:=\sqrt{r_{tt}}\approx 0.8128$. On the other hand, the Pearson
product moment correlation between $\bX^{(g)}$ and $\bX^{(h)}$ is
$r_{gh}\approx$0.9941. In other words, the split-half reliability is
$r_{gh}\approx 0.9941$.

We can now proceed to estimate the true scores following the
discussion presented in Section~\ref{sec:true_score}.  For the
observed score $X_i$, the estimated true score is
${\widehat{T}_i}=\beta_1 X_i + \alpha_1$ (as per Equation~\ref{eqn:te}
and discussions thereafter), where $\beta_1=r_{tt}$ and
$\alpha_1={\bar{X}}(1 - \beta_1)$.  The true scores can also be
estimated using the split-half reliability instead of the reliability
from the classical definition. Then the above regression coefficients
$\beta_q$ and $\alpha_q$ are computed as above, except this time,
$r_{tt}$ is replaced by $r_{gh}$, $q=1,2$. We present the true scores
estimated using both the reliability from the classical definition
(${\widehat{T}}_{classical}$ in black) as well as the split-half
reliability (${\widehat{T}}_{split-half}$ in red) in
Figure~\ref{figure:t_est}. In this figure, the errors in the estimated
true scores are superimposed as error bars on this estimate, in
respective colours. This error is considered to be $\pm S_E$, where
the error variance is $S_E^2 = (1- r_{tt})S_X^2$ when the classically defined reliability is implemented and $S_E^2 = (1- r_{gh})S_X^2$ when the split-half reliability is used.

In fact, the method of using $r_{gh}$ in the estimation of the true
scores will result in the true scores ${\widehat{T}}_{split-half}$
being over-estimated for $X>{\bar{X}}$ and under-estimated for $X <
{\bar{X}}$, as shown in Theorem~\ref{theorem:over}. This can be easily
corroborated in our results in Figure~\ref{figure:t_est}; the higher
value of $r_{gh}$ (than of $r_{tt}$) results in
${\widehat{T}}_{split-half} \geq {\widehat{T}}_{classical}$ for
$X>{\bar{X}}$ and in ${\widehat{T}}_{split-half} \leq
{\widehat{T}}_{classical}$ for $X<{\bar{X}}$.

Figure~\ref{figure:t_est} also includes the sample probability
distribution of the point estimate of the true score obtained from the
linear regression model suggested in Equation~\ref{eqn:te}, given the
observed scores and using $r_{tt}$ and $r_{gh}$, i.e. $\Pr(T \leq t\vert
r, \bX)$, where $r$ is the reliability.
We invoke the lookup table that gives the examinee indices
corresponding to a point estimate of the true score. Then the
percentile rank of all those examinees whose true score is (point)
estimated as $T$, is given as $100\times\Pr(T \leq t\vert r,
\bX)$. (This lookup table is not included in the text for brevity's
sake). For example, the true score
${\widehat{T}}_{split-half}\in[20,21)$ is estimated using $r=r_{gh}$,
  for examinees with indices 893, 867, 210, 837, 834, 408, 706, 690,
  655, 653, 161, 638, 312, 308, 149, 290, 539, 260. Then the
  percentile rank of all these examinees is about 93.
 
\begin{figure}[!t]
\centering
\includegraphics[width=14cm]{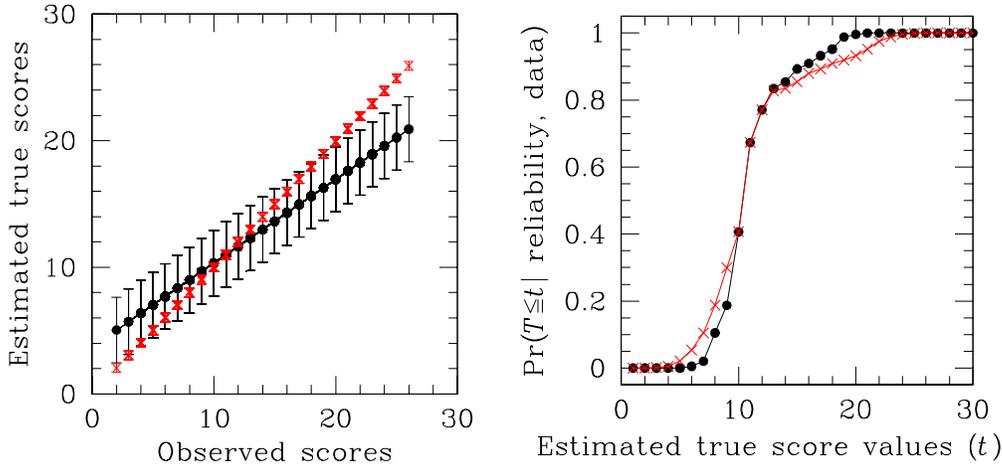}
\vspace*{-.3cm}
\caption{\small{{\it Left}: figure showing the true scores estimated at observed examinee scores of the real test data DATA-I. The true scores estimated using the reliability computed from the classical definition is shown in black filled circles; the corresponding estimates of the error scores are superimposed as error bars. The true scores estimated using the split-half reliability are plotted in red crosses, as are the corresponding error scores. {\it Right}: the sample cumulative probability distributions of point estimate of true scores obtained from the linear regression model given in Equation~\ref{eqn:te} where implementation of the classically defined reliability results in the plot in black circles while that obtained using split-half reliability is in red crosses. This cumulative distribution can be implemented to identify the percentile rank of an examinee, in conjuction with the lookup table containing the ranking of examinee indices by the point estimate of the true score.}}
\label{figure:t_est}
\end{figure}

Thus, our method of splitting the test into 2 parallel halves helps to
find a unique measure of reliability as per the classical definition,
for a given real data set. Using this we can then estimate true scores
for each observed score in the data.

\subsection{Weighted battery score using real data}
\label{sec:weightedresults}
\noindent
In order to illustrate our method of finding weights relevant to the
computation of the classically defined weighted reliability of a test
battery (Equation~\ref{eqn:rel_bat_wt}), we employ real life test data sets
DATA-II(a) and DATA-II(b). These data sets comprise the examinee scores of
784 examinees who took 2 tests (aimed at selection for a position). The
test that resulted in DATA-II(a) aimed at measuring examinee 
verbal ability, while the test for which DATA-II(b) was the result, measured
ability to interpret data. The former test contained 18 items and the latter 22
items. Histograms of examinee scores and item scores of these 2 tests
are shown in Figure~\ref{figure:histos_batt}. The mean and variance of
the 2 tests corresponding to DATA-II(a) and DATA-II(b) are
approximately 4.30, 6.25 and 4.20, 4.06 respectively. The battery
comprising these two tests is considered for its reliability.

\begin{figure}[!h]
     \centering
\includegraphics[width=18cm]{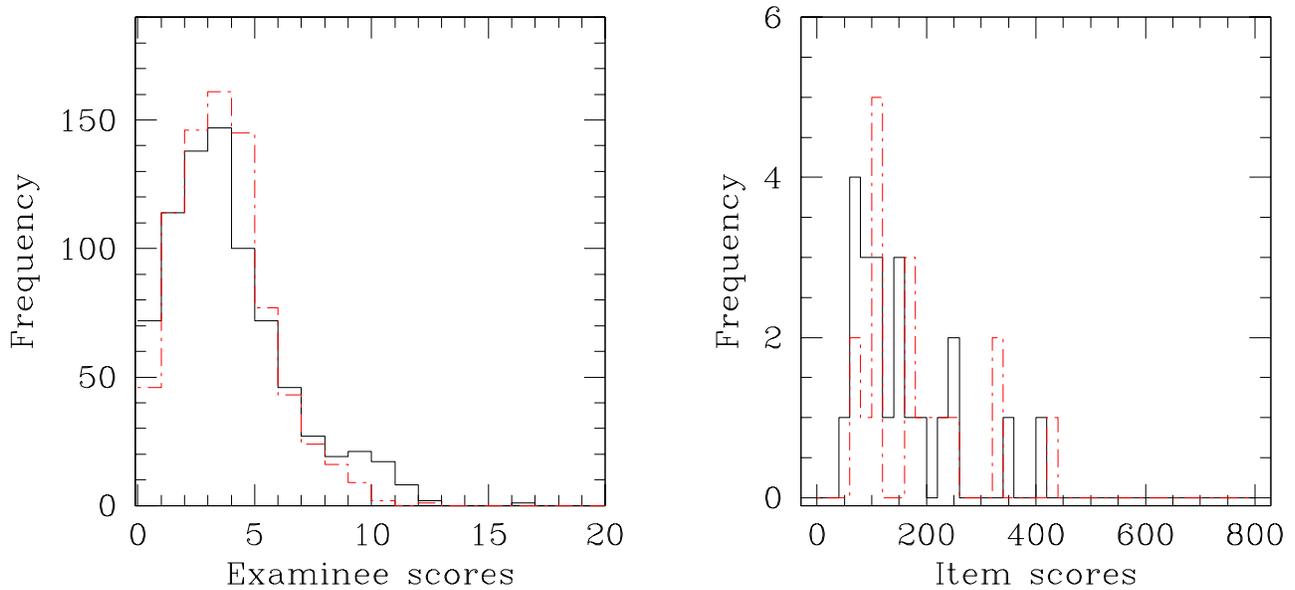}
\vspace*{-1cm}
\caption{\small{{\it Left:} figure showing histograms of the examinee scores in the tests, the scores of which constitute data sets DATA-II(a) (in black) and DATA-II(b) (in red dotted line). {\it Right:} figure showing histograms of item scores in real data DATA-II(a) (in black solid lines) and in DATA-II(b) (in red broken lines).} }
\label{figure:histos_batt}
\end{figure}

We split each test into 2 parallel halves using our method of
dichotomisation delineated in Section~\ref{sec:dichot}. Splitting
DATA-II(b) results in 2 halves, the means of the examinee scores of
which are about 2.10 and the variances of the examinee scores are
about 1.80 and 1.65. Reliability defined classically
(Equation~\ref{eqn:4thstep}) computed using these split halves is
about 0.35. The correlation between the split halves is about 0.90. On
the other hand, splitting DATA-II(a) results in 2 halves, the means of
which are equal to about 2.15 and the variances of which are about
1.87 and 2.19. Classically defined reliability of this test using
these split halves, turns out be 0.66 approximately. The correlation
coefficient between the two halves is about 0.95. The true scores
estimated for both observed data sets are shown in
Figure~\ref{figure:batt}, with errors of estimation superimposed as
$\pm S_E$ where $S_E^2$ is the variance of the error scores that are
modelled as distributed as ${\cal N}(0, S_E^2)$.

\begin{figure}[!ht]
     \centering
\includegraphics[width=10cm]{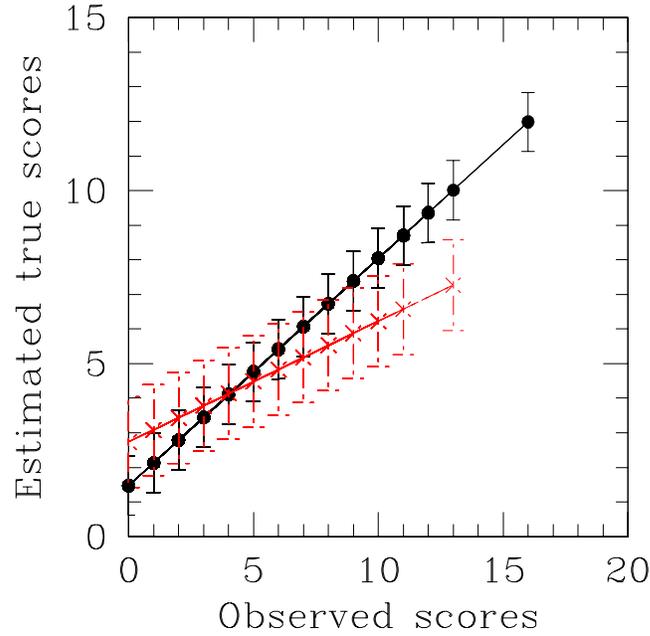}
\vspace*{-1cm}
\caption{\small{Figure showing true scores estimated using the observed scores in real data set DATA-II(a) (in black filled circles) and in DATA-II(b) (in red crosses). Errors of estimate are depicted as error bars that correspond to $\pm (1-r_{tt})^2S_X^2$ where the $r_{tt}$ is the classically defined reliability computed using the respective real data set.}}
\label{figure:batt}
\end{figure}

Using the scores in the 2 tests in the considered battery, namely sets DATA-II(a) and DATA-II(b), we compute the variance-covariance matrix $\bD$ of the observed scores (see Equation~\ref{eqn:de}). Then in this example, $\bD$ is a 2$\times$2 matrix, the diagonal elements of which are the variances of the 2 tests and the off-diagonal elements are the covariances of the examinee test score vectors $\bX_1$ and $\bX_2$ of these 2 tests. Then in this example real-life test battery, recalling that $\bee = \displaystyle{\left(1, 1\right)^T}$ we get\\
$$ \bD\approx\left( \begin{array}{cc}
6.24 & 2.46 \\
2.46 & 4.06 \end{array} \right),\quad\mbox{so that}\quad \bD^{-1}\bee \approx \displaystyle{\left(0.08, 0.20\right)^T},\quad \bee^T\bD^{-1}\bee\approx 0.28.$$
Then the weights are $\approx$ 0.7028 and 0.2972 (to 4 significant figures). 

Then recalling that (upto 4 decimal places) the reliability and test variance of the 2 tests the score data of which are DATA-II(a) and DATA-II(b) are 0.6571, 6.2405 and 0.3488, 4.0571 respectively, the covariance between the 2 tests is 2.4571, we use Equation~\ref{eqn:rel_bat_wt} to compute the reliability of this real battery to be approximately \\
$$\displaystyle{\frac{0.6571\times 6.2405\times 0.7028^2 + 
                      0.3488\times 4.0571\times 0.2972^2 + 
                      2\times 0.7028\times 0.2972\times 2.4571}
{0.7028^2\times 6.2405 + 0.2972^2\times 4.0571 + 2\times 0.7028\times 0.2972\times 2.4571}}\approx 0.7112.$$

Thus, using our method of splitting a test in two parallel halves, we
have been able to compute reliability of the test as per the classical
definition and extended this to compute the reliability of a real
battery comprising two tests. For the used real data sets DATA-II(a)
and DATA-II(b), the battery reliability turns out to be about 0.71.

\section{Summary $\&$ Discussions}
\label{sec:conclu}
\noindent
The paper presents a new, easily calculable split-half method of
achieving reliability of tests, using the classical definition where
the basic idea implemented is that the square of the magnitude of the
difference between the score vectors $\bX_g$ and $\bX_h$ of $N$
examinees in the $g$ and $h$ sub-tests obtained by splitting the full
test, is proportional to the variance $S_E^2$ of the errors in the
scores obtained by the examinees who take the test, i.e. $\parallel
\bX_g - \bX_h\parallel^2 = N S_E^2$. Here, working within the paradigm
of Classical Test Theory, the error in an examinee's score is the
difference between the observed and true scores of the examinee. Our
method of splitting the test is iterative in nature and has the
desirable properties that the sample distribution of the split halves
are nearly coincident, indicating the approximately equal means and
variances of the split halves. Importantly, the splitting method that
we use, ensures maximum split-half correlation between the split
halves and the splitting is performed on the basis of difficulty of
the items (or questions) of the test, rather than examinee
attributes. A crucial feature of this method of splitting is that the
splitting being in terms of item difficulty, the method requires very
low computational resources to split a very large test data set into
two nearly coincident halves. In other words, our method can easily be
implemented to find the classically defined reliability using test
data that is obtained by collating responses from a very large sample
of examinees, on whom a test of as large or as small a number of items
has been administered. In our demonstration of this method, a
moderately large real test on 912 examinees and 50 items, convergence
to optimum splitting (splitting into halves that share equal means and
nearly equal variances) was achieved in about 0.024
seconds. Implementation of sets of toy data, generated under different
model choices for examinee ability was undertaken: 999 examinees
responding to a 50-item test, as well as much larger cohorts of
examinees--500,000 to 50,000--taking a 50-item test, and a cohort of
1000 examinees taking 100 to 1000-item tests. The order of our
splitting algorithm is corroborated to be quadratic in half of the
number of items in the test, while computational time for reliability
computation (input+output times, in addition to splitting of the test)
varies linearly with examinee sample size, so that even for the
500,000 examinees taking the 50-item test, reliability is computed to
be less than 10 seconds. Once the reliability of the test is computed,
it is exploited to perform interval estimation of the true score of
each examinee, where the error of this estimation is modelled as the
test error variance.

Subsequent to the dichotomisation of the test, we invoke a simple
linear regression model for the true score of an examinee, given the
observed score $X$, to achieve an interval estimate of the true score,
where the interval is modelled as $\pm 1S_E$. We recognise this
interval to be in excess of $\pm 1$ standard deviation of the error of
estimation of the true score for a given $X$, as provided by the
regression model; in other words, our estimation of uncertainty on the
estimated true score is pessimistic.

This method of splitting a test into 2 parallel halves, forms the
basis of our computation of the reliability of a battery, i.e. a set
of tests, as per the classical definition. A weighted battery score is
used in this computation where we implement a new way for the
determination of the weights, by invoking a Lagrange multiplier based
solution. We illustrate the implementation of this method of
determining weights--and thereby of computing the reliability of a
test battery following the computation of the reliabilities of the
component test as per the classical definition--on a real test battery
that comprises 2 component tests.

We have presented a new method of computing reliability as per
the classical definition, and demonstrated its proficiency and
simplicity. Thus, it is possible to uniquely find test characteristics like
reliability or error variance from the data, and such can be adopted
while reporting the results of the administered test.  In this
paradigm, testing of hypothesis of equality of error variance from two
tests will help to compare the tests.




\renewcommand\baselinestretch{1.5}
\small
\bibliographystyle{ECA_jasa}


%




\end{document}